\documentclass[12pt]{article}
\usepackage[left=2.5cm,right=2.5cm,top=3cm,bottom=3cm]{geometry}
\usepackage[center]{titlesec}
\usepackage[english]{babel}
\usepackage{amsmath,amsthm,amssymb}
\usepackage{listings}
\usepackage{xcolor}
\usepackage{enumitem}
\usepackage{extarrows}
\usepackage{graphicx,subfigure,float}
\usepackage{bm}
\usepackage{mathrsfs}
\usepackage{booktabs}
\usepackage{multirow}
\usepackage{dsfont}
\usepackage{setspace}
\usepackage{latexsym,multicol,calligra}
\usepackage{pstricks,stackengine}
\usepackage{tikz,tikz-cd,tkz-graph}
\usepackage{quiver}

\begin{document}

\title{On Hamiltonian Structures of \\ Partial Difference Equations}
\author{Zhonglun Cao}
\maketitle

\begin{abstract}
We first introduce the notion of Hamiltonian structure for a partial difference equation. Then we construct some infinite quivers, and realize 
the discrete KdV equation, the Hirota-Miwa equation and its various reductions as the mutation relations of the corresponding cluster algebras.
Finally, we show that the log-canonical Poisson structures associated to these cluster algebras give the Hamiltonian structures or the
bihamiltonian structures of these partial difference equations.
\end{abstract}

\newtheorem{thm}{Theorem}[section]

\theoremstyle{definition}
\newtheorem{eg}{Example}[section]

\theoremstyle{definition}
\newtheorem{definition}{Definition}[section]

\theoremstyle{remark}
\newtheorem*{remark}{Remark}

\section{Introduction}

The Hamiltonian structures and bihamiltonian structures of an integrable partial differential equation play an important role in the study of
the integrability of these equations. Recently, Boris Dubrovin, Youjin Zhang and their collaborators established a classification
theory for partial differential equations possessing semisimple bihamiltonian structures\cite{donagi1996geometry,dubrovin2001normal},\cite{supertaucover}-\cite{liu2021linearization}. They showed that, to each semisimple cohomological field
theory (CohFT), one can associate a hiearchy of integrable partial differential equations, which possess many nice properties, including the semisimple
bihamiltonian structures, such that the partition function of this field theory is a tau function of this integrable hierarchy. 

It is interesting that, for certain CohFT, the Dubrovin-Zhang hierarchy actually does not consist of partial \emph{differential} equations but
partial \emph{difference} equations with continuous time variable and discrete spatial variables, such as the Toda hierarchy, which corresponds
to the Gromov-Witten theory of the projective line. In these cases, we can still treat the partial difference equation as a vector field on
an infinite-dimensional manifold, and talk about its Hamiltonian structures or bihamiltonian structures.

On the other hand, there are many important integrable partial difference equation whose time variable and spatial variable are both discrete\cite{hirota1977nonlinear1,date1982method}.
These kind of integrable systems usually describe certain discrete symmetries of a system, and it is not easy or possible to find out the
infinitesimal generators of these discrete symmetries, so one cannot talk about its Hamiltonian structures. 

In the present paper, we try to give a definition of Hamiltonian structures for a partial difference equation with both discrete time variable
and spatial variable. Then, to illustrate that our definition is reasonable, we consider some examples, which are closely related to cluster algebra,
and show that these partial difference equation do possess Hamiltonian structures in our sense, by using some results in the theory of the
cluster algebra.

Cluster algebra theory was founded by the series of papers by Fomin and Zelevinsky, together with Berenstein\cite{fomin2002cluster,fomin2003cluster,berenstein2005cluster,fomin2007cluster}. A cluster algebra is defined by seeds consisting of cluster variables, coefficients and skew-symmetrizable (integer) matrices. There is an operation called seed mutation defined on the seeds. If we only consider the cluster algebra without coefficients, it can be defined by quivers (directed graphs) whose vertices correspond to variables on which an operation called quiver mutation acts. The Poisson structure of cluster algebras was studied by Gekhtman, Shapiro and Vainshtein\cite{gekhtman2003cluster}, there exists a log-canonical Poisson bracket on cluster variables and it is mutation compatible. It is known that some difference equations can be connected with periodic cluster algebras by the formula of variables mutation. Inoue and Nakanishi extended the Poisson structure to some difference equation that related to cluster algebras \cite{inoue2011difference}. We will consider more examples of partial difference equations that are rerlated to infinite quivers and cluster algebras, and show that the log-canonical Poisson structure of the cluster algebra gives the Hamiltonian structures
of the corresponding partial difference equation.

Our definition also works for the ordinary difference equations which correspond to finite quivers. Their Poisson structure were studied by Hone and collaborators\cite{fordy2014discrete,hone2017reductions,hone2019cluster}.

This paper is organized as follows: We first introduce the notion of Hamiltonian structure of difference equations in section 2;
then we recall the definition of cluster algebra and Poisson structure of it in section 3. In section 4, we consider some specific quivers corresponding to the discrete KdV equation and calculate the general form of the Hamiltonian structure of the discrete KdV equation. Finally in section 5, we give some other reductions of Hirota-Miwa equation and the Hamiltonian structures of them.

\section{Hamiltonian structures of difference equations}

A difference equation contains several unknown functions $u=(u^1, \dots, u^m)$, each of which depends on several independent variables
$n=(n_1, \dots, n_d)\in\mathbb{Z}^d$. We introduce the shift operator for the $i$-th independent variable
\[\Lambda_i(u(n_1, \dots, n_d)):=u(n_1, \dots, n_{i}+1, \dots, n_d),\]
and denote its $p$-power as $\Lambda_i^p$ for $p\in\mathbb{Z}$, then a typical difference equation takes the following form
\begin{equation}
Z(n):=Z(u(n),\Lambda_{\cdots}^{\pm 1}(u(n)),\Lambda_{\cdots}^{\pm 2}(u(n)),\dots)=0. \label{dc-1}
\end{equation}
This equation should hold for every $n\in\mathbb{Z}^d$, so for an $h\in \mathbb{Z}^d$, we have
\[Z(n+h)=Z(u(n+h),\Lambda_{\cdots}^{\pm 1}(u(n+h)),\Lambda_{\cdots}^{\pm 2}(u(n+h)),\dots)=0.\]
Hence the difference equation \eqref{dc-1} is \emph{translation invariant}.

Denote by $\mathcal{R}$ the polynomial ring of all symbols $\{u(n)\mid n\in\mathbb{Z}^d\}$ over a certain field $\mathbb{K}$, then $\Lambda_i \ (i=1, \dots, d)$ can be
extended as automorphisms on $\mathcal{R}$. An ideal $J$ of $\mathcal{R}$ is called translation invariant, if for all $i=1, \dots, d$, and 
$p\in\mathbb{Z}$, we have $\Lambda_i^p(J)\subseteq J$. Then for a system of difference equation $Z=(Z_1, \dots, Z_n)$ such that each
$Z_i$ takes the form \eqref{dc-1}, we can defined the translation invariant ideal $J_Z$ generated by $Z$. 

A Poisson bracket on $\mathcal{R}$ is defined as usual, that is a bilinear operation
\[\{\ ,\}:\mathcal{R}\times\mathcal{R}\to\mathcal{R},\]
satisfying the skew-symmetric condition, the Leibniz rule, and the Jacobi identity. A Poisson bracket is called \emph{translation invariant},
if for any $f, g\in\mathcal{R}$, and any shift operator $\Lambda_i^p$ we have
\begin{equation}\label{traninvariant}
	\{\Lambda_i^p(f), \Lambda_i^p(g)\}=\Lambda_i^p(\{f,g\}).
\end{equation}

\begin{definition}\label{definition}
For a system of difference equation $Z$, and a translation invariant Poisson bracket $\{\ ,\}$, if the translation invariant ideal $J_Z$ generated
by $Z$ is also an ideal of the Lie algebra $(\mathcal{R}, \{\ ,\})$, that is $\{J_Z, \mathcal{R}\}\subseteq J_Z$, then we say that
$\{\ ,\}$ is a Hamiltonian structure of the discrete system $Z$.
\end{definition}

We can also defined the notion of bihamiltonian structure for a discrete system.
\begin{definition}
For a system of difference equation $Z$, and two translation invariant Poisson bracket $\{\ ,\}_1$ and $\{\ ,\}_2$, if for any constant 
$\lambda\in\mathbb{K}$, the linear combination $\{\ ,\}_\lambda=\{\ ,\}_2-\lambda \{\ ,\}_1$ is always a Hamiltonian structure of $Z$,
then we say that the pair $(\{\ ,\}_1, \{\ ,\}_2)$ is a bihamiltonian structure of the discrete system $Z$.
\end{definition}

\section{Cluster algebras}

In this section, we recall the notions of cluster algebra introduced by Fomin and Zelevinsky\cite{fomin2002cluster}. One can see the relation between cluster algebras and the difference equations. We also recall the Poisson structures corresponding to cluster algebras.

\subsection{Quiver and cluster algebra}
Let $I \subset \mathbb{Z}$ be an index set and $B=(b_{ij})_{i,j\in I}$ be a skew-symmetric integer matrix. When $I$ is infinite, we always assume that an infinite dimensional skew-symmetric matrix B has only finitely many nonzero elements in each row and in each column.

\begin{definition}[matrix mutation]
For a skew-symmetric matrix $B$ and $k \in I$, we have the matrix mutation $B' := \mu_{i}(B)$ of $B$ at index $k$ defined by

$$
b_{i j}^{\prime}= \begin{cases}-b_{i j} & i=k \text { or } j=k \\ b_{i j}+b_{i k}\left[b_{k j}\right]_{+}+\left[-b_{i k}\right]_{+} b_{k j} & i, j \neq k .\end{cases}
$$
where $[a]_+ := \max\{a,0\}$.
\end{definition}

Any skew-symmetric matrix can be represented by a quiver. A quiver is a directed graph, there are two structures called a loop and a 2-cycle, respectively. 

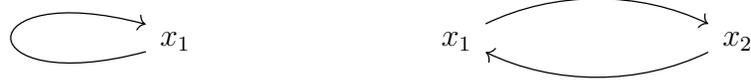
\begin{equation*}
    \begin{tikzcd}
	\arrow[loop left, distance=6em, start anchor={[yshift=-1ex]west}, end anchor={[yshift=1ex]west}]{}{}x_1&&&{x_1} &&& {x_2}
	\arrow[curve={height=-18pt}, from=1-4, to=1-7]
	\arrow[curve={height=-18pt}, from=1-7, to=1-4]
    \end{tikzcd}\
\end{equation*}

For a skew-symmetric matrix $B$ with an index set $I$, one can associate a quiver $Q(B)$ without loops and 2-cycles by the following rules:
\begin{itemize}
    \item Each index $i$ corresponds to a vertex $x_i$.
    \item If and only if $b_{ij} > 0$, draw $b_{ij}$ arrows from the vertex $x_i$ to the vertex $x_j$.
\end{itemize}
Since $b_{ii} = 0$ for any index $x_i$, there is no loops. If $b_{ij} > 0$, we have $b_{ji} < 0$, so there is no 2-cycles.

Conversely, from a quiver without loops and 2-cycles, one can apply the above rules in the opposite direction and get a corresponding skew-symmetric matrix B.
It is clear that this correspondence is one-to-one.

Consider a quiver without loops or 2-cycles, denoted by $Q$. There also is a corresponding mutation operation for $Q$ by translating the matrix mutation. 

\begin{definition}[quiver mutation]
For any vertex $x_i$ in $Q$, the mutation operation $\mu_i$ at $x_i$ is defined as follows:
\begin{itemize}
    \item For each pair $(j,k)$, such that $j,k\neq i $, if there are $p > 0$ arrows from the vertex $x_j$ to the vertex $x_i$, and $q > 0$ arrows from the vertex $x_i$ to the vertex $x_k$, draw $pq$ new arrows from $x_j$ to $x_k$.
    \item Remove remove all 2-cycles in the quiver thus obtained.
    \item Invert all arrows into and out of the vertex $x_i$.
\end{itemize}
The new quiver is denoted by $Q':=\mu_i(Q)$.

\end{definition}

Let's introduce some more fundamental notions of the cluster algebras.
\begin{definition}[Seed, cluster variable]
    A seed $(Q,\{x_i\})$ is a pair consisting of a quiver $Q$ and the set $\{x_i\}_{i \in I}$ of vertices in $Q$. Each vertex in the quiver of a seed is called a cluster variable.
\end{definition}

\begin{definition}[Seed mutation]
    For any seed $(Q,\{x_i\})$ and a vertex $x_k$, define the seed mutation $\mu_k$ at $x_k$ as follows:
    \begin{itemize}
        \item $\mu_k(Q)$ is just the quiver mutation at $x_k$ of quiver $Q$.
        \item 
    $\begin{aligned}
\mu_k(x_i) := \begin{cases}\frac{\prod_{j \leftarrow k} x_j+\prod_{j \rightarrow k} x_j}{x_k} & i=k, \\
x_i & i \neq k.\end{cases}
\end{aligned}$
    \end{itemize}
The above symbol $\prod_{j \leftarrow k} x_j( \text{resp.} \prod_{j \rightarrow k} x_j)$ means the product of all vertices in $Q$ which have arrows from (resp. to) the vertex $x_k$.
\end{definition}
The new seed is denoted by $(Q',\{x_i^{\prime}\})$. It is clear that the seed mutation is involutive, $i.e., \mu_k^2 =$id. 

For an initial seed $(Q,\{x_i\})$, there are $|I|$ (the number of the element of the index set $I$, it can be infinite) mutations $\{\mu_i\}_{i \in I}$ and $|I|$ new seeds. Each new seed also has as many as mutations. By iterating mutations from the initial seed $(Q,\{x_i\})$, we define $\mathcal{X}$ to be the set of all obtained cluster variables.
\begin{definition}[Cluster algebra]
   The cluster algebra (without coefficients) $\mathcal{A}(Q,\{x_i\})$ is a $\mathbb{Z}$-subalgebra generated by $\mathcal{X}$.
\begin{equation*}
\mathcal{A}(Q,\{x_i\})=\mathbb{Z}[x \mid x \in \mathcal{X}] \subset \mathbb{Q}(x \mid x \in \mathcal{X}).
\end{equation*}
\end{definition}


For simplicity, we introduce an example with the finite index to illustrate the connection between the quiver and the difference equation. 
\begin{eg}[Somos-4]
Somos-4 is an ordinary difference equation: 
\begin{equation*}
    x_{n} x_{n+4}=x_{n+1} x_{n+3}+x_{n+2}^2.
\end{equation*}
The quiver $G$ corresponding to somos-4 is defined as:
\begin{equation}\label{qui:somos4}
    \begin{tikzcd}[ampersand replacement=\&]
    	{x_{1}} \& {x_{2}} \& {x_{3}} \& {x_{4}}
	\arrow[from=1-2, to=1-1]
	\arrow[from=1-4, to=1-3]
	\arrow["3"{description}, from=1-3, to=1-2]
	\arrow["2"{description}, curve={height=18pt}, from=1-1, to=1-3]
	\arrow["2"{description}, curve={height=18pt}, from=1-2, to=1-4]
	\arrow[curve={height=24pt}, from=1-4, to=1-1]
        \end{tikzcd}
\end{equation}

We can view $x_{5}$ as a cluster variable obtained by the mutation at $x_1$ of $G$ since 
\begin{equation*}
    \mu_1(x_1)=\frac{x_{2} x_{4}+x_{3}^2}{x_{1}}=x_5.
\end{equation*}
Then $G'= \mu_1(G)$ is 
\[\begin{tikzcd}[ampersand replacement=\&]
    	{x_1'} \& {x_2} \& {x_3} \& {x_4}
	\arrow[from=1-1, to=1-2]
	\arrow["3"{description}, from=1-4, to=1-3]
	\arrow[from=1-3, to=1-2]
	\arrow["2"{description}, curve={height=-18pt}, from=1-3, to=1-1]
	\arrow["2"{description}, curve={height=18pt}, from=1-2, to=1-4]
	\arrow[curve={height=-24pt}, from=1-1, to=1-4]
\end{tikzcd}\]
At last, we replace $x_1'$ with $x_5$, and move it to the tail of the queue:
\[\begin{tikzcd}[ampersand replacement=\&]
    	{x_2} \& {x_3} \& {x_4} \& {x_5}
	\arrow[from=1-2, to=1-1]
	\arrow[from=1-4, to=1-3]
	\arrow["3"{description}, from=1-3, to=1-2]
	\arrow["2"{description}, curve={height=18pt}, from=1-1, to=1-3]
	\arrow["2"{description}, curve={height=18pt}, from=1-2, to=1-4]
	\arrow[curve={height=24pt}, from=1-4, to=1-1]
\end{tikzcd}\]
Notice that this new quiver is just $G$, except replacing each $x_n$ with $x_{n+1}$. By iterating the above process we can conclude that the variables in somos-4 consist a subset of the cluster algebra $\mathcal{A}(G,\{x_i|i=1,2,3,4\})$. Hence some good properties of the cluster algebra can be proved in somos-4, such as Laurent phenomenon.
\end{eg}

\subsection{Poisson structure of cluster algebra}

There are some results about the Poisson structure of cluster algebras and difference equations \cite{gekhtman2003cluster}\cite{inoue2011difference}. We are especially concerned about the case with an infinite index set. 

Fix a skew-symmetric matrix $B=(b_{ij})$ with the index set $I$. The corresponding seed is $(Q,\{x_i\})$.
\begin{definition}[Log-canonical Poisson bracket]
The log-canonical Poisson bracket for $\{x_i\}$ is defined by
\begin{equation}
\label{equ:poisson1}
    \{x_i,x_j\}=p_{ij} x_i x_j,\quad p_{ij} \in \mathbb{Q}
\end{equation}
satisfying
\begin{itemize}
    \item $skew-symmetric: \{x_i,x_j\}=-\{x_j,x_i\}, i.e., p_{ij}=-p_{ji},$\\
    \item $Jacobi-identity: \{x_i,\{x_j,x_k\}\}+\{x_j,\{x_k,x_i\}\}+\{x_k,\{x_i,x_j\}\}=0$
\end{itemize}
\end{definition}

We want to study the log-canonical Poisson bracket which is compatible
with the seed mutation in the following sense:
\begin{definition}[mutation compatible]
    We say a Poisson bracket for $\{x_i\}_{i \in I}$ is mutation
compatible if, for any index $k$, the bracket induced for $\{x_i'\}=\{\mu_k(x_i)\}$ again has the form $ \{x_i',x_j'\}=p_{ij}' x_i' x_j'$.
\end{definition}

\begin{thm}\cite{inoue2011difference}\label{thm:pmutation}
(i)For a skew-symmetric matrix $P=(p_{ij})_{i,j \in I}$, the corresponding Poisson bracket (\ref{equ:poisson1}) is mutation compatible if and only if $PB$ is a diagonal matrix.

(ii) Suppose that $PB$ is a diagonal matrix. Let $P'=(p_{ij}')_{i,j \in I}$ be the matrix for the induced Poisson bracket $ \{x_i',x_j'\}=p_{ij}' x_i' x_j'$. Then $P'$ is given by
\begin{equation}\label{equ:pmutation}
p_{i j}^{\prime}= \begin{cases}-p_{i k}+\sum_{l: b_{l k}>0} b_{l k} p_{i l} & i \neq j=k, \\ -p_{k j}+\sum_{l: b_{l k}>0} b_{l k} p_{l j} & k=i \neq j, \\ p_{i j} & \text { otherwise. }\end{cases}
\end{equation}
\end{thm}

We can view the formula (\ref{equ:pmutation}) as the Poisson bracket mutation.
\begin{definition}[Poisson structure along the mutation]
Suppose that $P$ define a log-canonical Poisson bracket for $\{x_i\}_{i \in I}$ and it is mutation compatible. 
We say there is a Poisson structure along the mutation if for any seed which can be obtained by finitely many mutations from the initial seed, there exists a mutation compatible log-canonical Poisson bracket which can be obtained in the same mutation way. The Poisson bracket matrix $P$ of the initial seed is called the initial Poisson matrix.
\end{definition}
Notice that the Poisson structure along the mutation doesn't always exist. According theorem \ref{thm:pmutation}, we need that $P'B'$ is still a diagonal matrix.

Now consider the case that $I$ is an infinite index set. For a pair of matrices $M$ and $N$ with the infinite index set $I$, we say $M$ is a left inverse of $N$ if $MN = id_{I}$.

\begin{thm}\cite{inoue2011difference}\label{thm:pb0}
Suppose that $I$ is an infinite index set, $B$ is a skew-symmetric matrix with index set $I$ and $B$ has no left inverse. There exists a nontrivial Poisson structure along the mutation if there is a nontrivial skew-symmetric matrix $P$ such that $PB =O$. Furthermore, for each pair $(B',P')$ obtained by mutation, we have $P'B' =O$.
\end{thm}

\section{Discrete KdV equation}

The discrete KdV equation is a well-known discrete integrable system. We will study the Hamiltonian structure of this equation in our sense by using some results about cluster algebra. 

\subsection{Quivers of the discrete KdV equation}

The discrete KdV equation is a two dimensional difference equation:
\begin{equation}\label{equ:kdv}
    \sigma_{l+1}^{n+1}=\frac{\sigma_{l+1}^n \sigma_{l-1}^{n+1}+\sigma_l^{n+1} \sigma_l^n}{\sigma_{l-1}^n}.
\end{equation}
Obviously, it can correspond some quiver and have the phenomenon just like somos-4. Okubo constructs a desired infinite quiver $G$\cite{okubo2013discrete}:

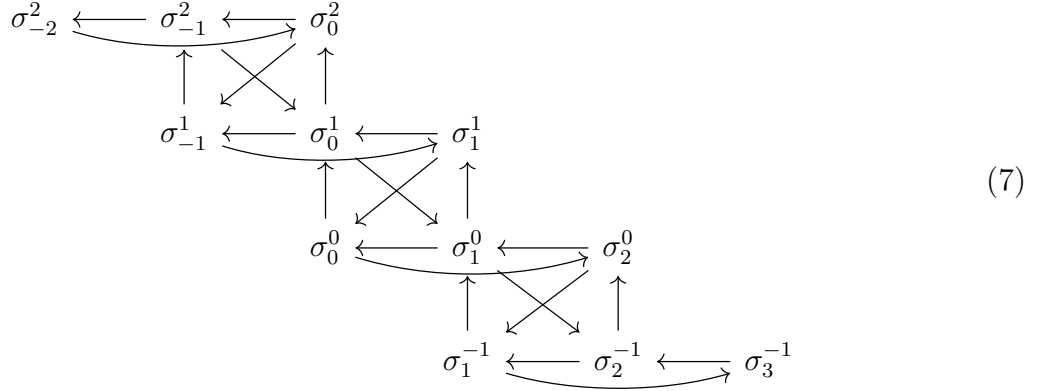
\begin{equation}\label{quiver:dkdv1}
    \begin{tikzcd}
	{\sigma_{-2}^{2}} & {\sigma_{-1}^{2}} & {\sigma_{0}^{2}} \\
	& {\sigma_{-1}^{1}} & {\sigma_{0}^{1}} & {\sigma_{1}^{1}} \\
	&& {\sigma_{0}^{0}} & {\sigma_{1}^{0}} & {\sigma_{2}^{0}} \\
	&&& {\sigma_{1}^{-1}} & {\sigma_{2}^{-1}} & {\sigma_{3}^{-1}}
	\arrow[from=3-3, to=2-3]
	\arrow[from=2-3, to=1-3]
	\arrow[from=2-3, to=2-2]
	\arrow[from=2-4, to=2-3]
	\arrow[from=3-4, to=2-4]
	\arrow[from=4-4, to=3-4]
	\arrow[from=3-4, to=3-3]
	\arrow[from=3-5, to=3-4]
	\arrow[from=4-5, to=3-5]
	\arrow[from=4-5, to=4-4]
	\arrow[from=2-2, to=1-2]
	\arrow[from=1-3, to=1-2]
	\arrow[from=1-2, to=1-1]
	\arrow[from=4-6, to=4-5]
	\arrow[from=3-5, to=4-4]
	\arrow[curve={height=12pt}, from=3-3, to=3-5]
	\arrow[curve={height=12pt}, from=2-2, to=2-4]
	\arrow[curve={height=12pt}, from=4-4, to=4-6]
	\arrow[curve={height=12pt}, from=1-1, to=1-3]
	\arrow[from=1-3, to=2-2]
	\arrow[from=1-2, to=2-3]
	\arrow[from=2-4, to=3-3]
	\arrow[from=2-3, to=3-4]
	\arrow[from=3-4, to=4-5]
\end{tikzcd}
\end{equation}

Inspired by the discrete KdV equation (\ref{equ:kdv}), we draw the quiver after mutation at vertex $\sigma_{0}^{0}$, replace ${\sigma_{0}^{0}}'$ by $\sigma_{2}^{1}$ and move it to the right position:

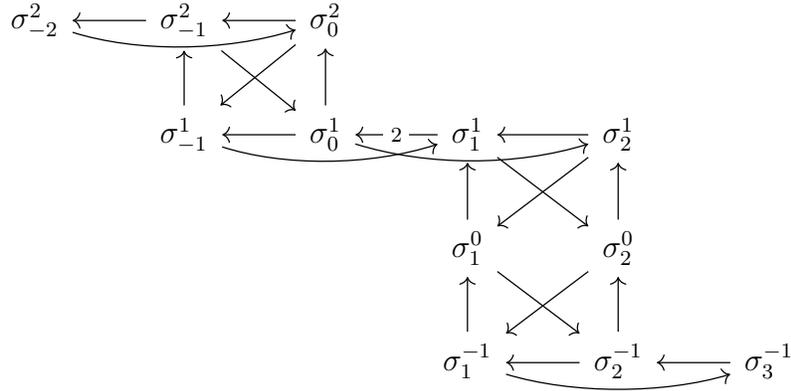
\[\begin{tikzcd}
	{\sigma_{-2}^{2}} & {\sigma_{-1}^{2}} & {\sigma_{0}^{2}} \\
	& {\sigma_{-1}^{1}} & {\sigma_{0}^{1}} & {\sigma_{1}^{1}} & {\sigma_{2}^{1}} \\
	&&& {\sigma_{1}^{0}} & {\sigma_{2}^{0}} \\
	&&& {\sigma_{1}^{-1}} & {\sigma_{2}^{-1}} & {\sigma_{3}^{-1}}
	\arrow[from=2-3, to=1-3]
	\arrow[from=2-3, to=2-2]
	\arrow["2"{description}, from=2-4, to=2-3]
	\arrow[from=3-4, to=2-4]
	\arrow[from=4-4, to=3-4]
	\arrow[from=4-5, to=3-5]
	\arrow[from=4-5, to=4-4]
	\arrow[from=2-2, to=1-2]
	\arrow[from=1-3, to=1-2]
	\arrow[from=1-2, to=1-1]
	\arrow[from=4-6, to=4-5]
	\arrow[from=3-5, to=4-4]
	\arrow[curve={height=12pt}, from=2-2, to=2-4]
	\arrow[curve={height=12pt}, from=4-4, to=4-6]
	\arrow[curve={height=12pt}, from=1-1, to=1-3]
	\arrow[from=1-3, to=2-2]
	\arrow[from=1-2, to=2-3]
	\arrow[from=3-4, to=4-5]
	\arrow[from=2-5, to=2-4]
	\arrow[curve={height=12pt}, from=2-3, to=2-5]
	\arrow[from=2-5, to=3-4]
	\arrow[from=3-5, to=2-5]
	\arrow[from=2-4, to=3-5]
\end{tikzcd}\]
We can see that the quiver after one mutation is no longer the initial one. But if we compose the mutation at $\{\sigma_{l}^{n} \in G | l+n=0\}$ and denote it by $\tilde{\mu_{0}}$, we have the new quiver $\tilde{\mu_{0}}(G)$:

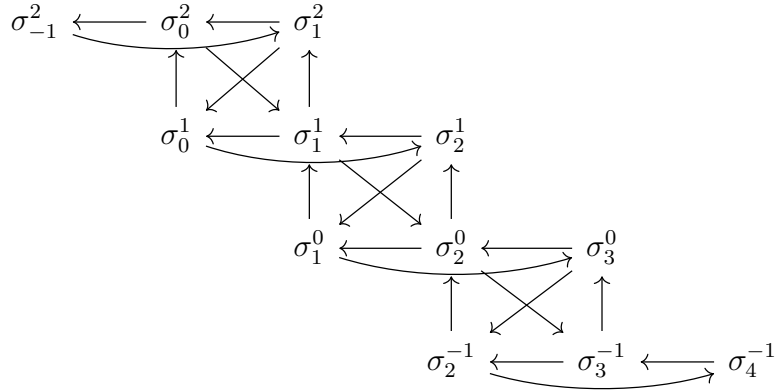
\[\begin{tikzcd}
	{\sigma_{-1}^{2}} & {\sigma_{0}^{2}} & {\sigma_{1}^{2}} \\
	& {\sigma_{0}^{1}} & {\sigma_{1}^{1}} & {\sigma_{2}^{1}} \\
	&& {\sigma_{1}^{0}} & {\sigma_{2}^{0}} & {\sigma_{3}^{0}} \\
	&&& {\sigma_{2}^{-1}} & {\sigma_{3}^{-1}} & {\sigma_{4}^{-1}}
	\arrow[from=2-2, to=1-2]
	\arrow[from=2-3, to=2-2]
	\arrow[from=3-3, to=2-3]
	\arrow[from=4-4, to=3-4]
	\arrow[from=1-2, to=1-1]
	\arrow[from=4-5, to=4-4]
	\arrow[from=2-4, to=2-3]
	\arrow[curve={height=12pt}, from=2-2, to=2-4]
	\arrow[from=2-4, to=3-3]
	\arrow[from=3-4, to=2-4]
	\arrow[from=2-3, to=3-4]
	\arrow[from=1-3, to=1-2]
	\arrow[curve={height=12pt}, from=1-1, to=1-3]
	\arrow[from=2-3, to=1-3]
	\arrow[from=1-2, to=2-3]
	\arrow[from=1-3, to=2-2]
	\arrow[from=3-5, to=3-4]
	\arrow[from=3-4, to=3-3]
	\arrow[from=4-5, to=3-5]
	\arrow[from=3-4, to=4-5]
	\arrow[from=3-5, to=4-4]
	\arrow[from=4-6, to=4-5]
	\arrow[curve={height=12pt}, from=4-4, to=4-6]
	\arrow[curve={height=12pt}, from=3-3, to=3-5]
\end{tikzcd}\]

$\tilde{\mu_{0}}(G)$ has the same structure as $G$, except that the horizontal coordinate of each variable adds one. Hence we say the quiver (\ref{quiver:dkdv1}) is a quiver corresponding to the discrete KdV equation.

\begin{remark}
    We can also compose the mutation at $\{\sigma_{l}^{n} \in G | l+n=2\}$ of $G$, then we will obtain the quiver where the horizontal coordinate of each variable minus one. 
\end{remark}

Recall the Hirota-Miwa equation:
\begin{equation}\label{equ:hirota-miwa}
        \sideset{^{n}}{^{m+1}_{l}}{\mathop{\tau}}  \sideset{^{n+1}}{^{m}_{l+1}}{\mathop{\tau}}=\sideset{^{n+1}}{^{m}_{l}}{\mathop{\tau}}  \sideset{^{n}}{^{m+1}_{l+1}}{\mathop{\tau}}+\sideset{^{n}}{^{m}_{l+1}}{\mathop{\tau}} \sideset{^{n+1}}{^{m+1}_{l}}{\mathop{\tau}}.
\end{equation}
Consider the quiver (\ref{quiver:hirota-miwa1})\cite{okubo2013discrete}. The variable mutation formula at any point in $\{\sideset{^{n}}{^{m}_{l}}{\mathop{\tau}}\vert l+m+n=-1 \}$ is the equation (\ref{equ:hirota-miwa}). After the composition of mutations at $\{\sideset{^{n}}{^{m}_{l}}{\mathop{\tau}}\vert l+m+n=-1 \}$, the resulting quiver has the same structure as (\ref{quiver:hirota-miwa1}). Hence the quiver (\ref{quiver:hirota-miwa1}) is a quiver corresponding to Hirota-Miwa equation.

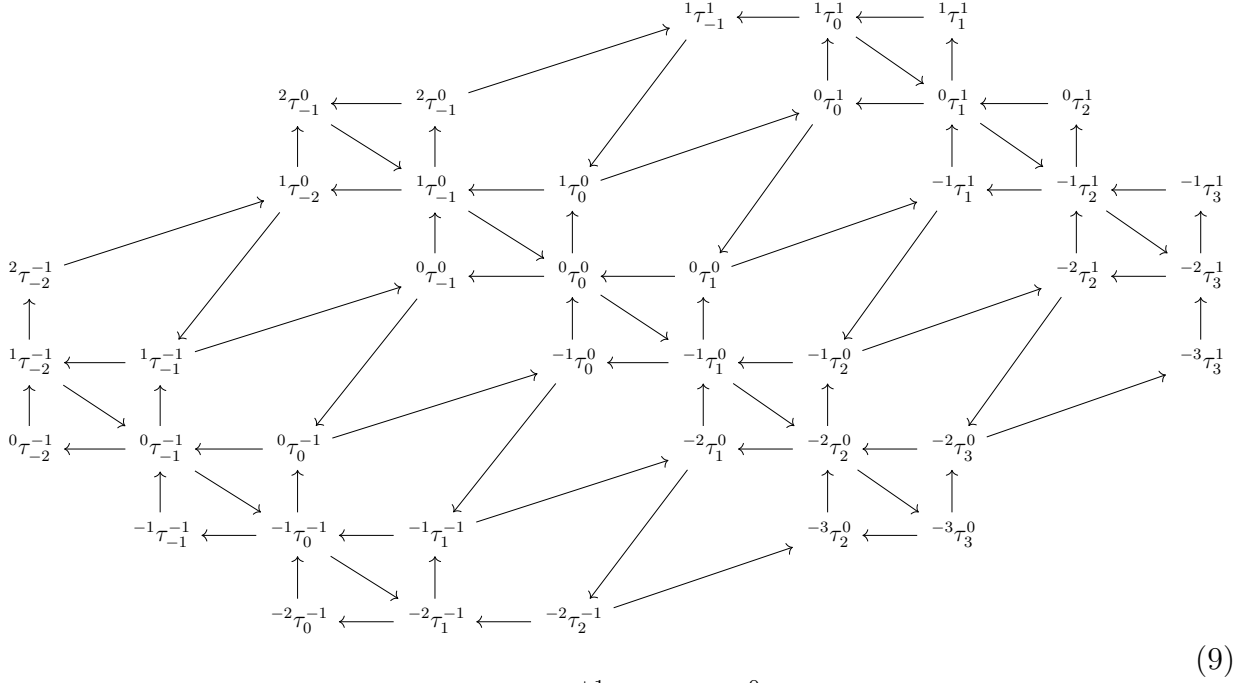
\begin{equation}\label{quiver:hirota-miwa1}
    \resizebox{\textwidth}{!}{
\begin{tikzcd}[ampersand replacement=\&]
	\&\&\&\&\& {\sideset{^{1}}{^{1}_{-1}}{\mathop{\tau}}} \& {\sideset{^{1}}{^{1}_{0}}{\mathop{\tau}}} \& {\sideset{^{1}}{^{1}_{1}}{\mathop{\tau}}} \\
	\&\& {\sideset{^{2}}{^{0}_{-1}}{\mathop{\tau}}} \& {\sideset{^{2}}{^{0}_{-1}}{\mathop{\tau}}} \&\&\& {\sideset{^{0}}{^{1}_{0}}{\mathop{\tau}}} \& {\sideset{^{0}}{^{1}_{1}}{\mathop{\tau}}} \& {\sideset{^{0}}{^{1}_{2}}{\mathop{\tau}}} \\
	\&\& {\sideset{^{1}}{^{0}_{-2}}{\mathop{\tau}}} \& {\sideset{^{1}}{^{0}_{-1}}{\mathop{\tau}}} \& {\sideset{^{1}}{^{0}_{0}}{\mathop{\tau}}} \&\&\& {\sideset{^{-1}}{^{1}_{1}}{\mathop{\tau}}} \& {\sideset{^{-1}}{^{1}_{2}}{\mathop{\tau}}} \& {\sideset{^{-1}}{^{1}_{3}}{\mathop{\tau}}} \\
	{\sideset{^{2}}{^{-1}_{-2}}{\mathop{\tau}}} \&\&\& {\sideset{^{0}}{^{0}_{-1}}{\mathop{\tau}}} \& {\sideset{^{0}}{^{0}_{0}}{\mathop{\tau}}} \& {\sideset{^{0}}{^{0}_{1}}{\mathop{\tau}}} \&\&\& {\sideset{^{-2}}{^{1}_{2}}{\mathop{\tau}}} \& {\sideset{^{-2}}{^{1}_{3}}{\mathop{\tau}}} \\
	{\sideset{^{1}}{^{-1}_{-2}}{\mathop{\tau}}} \& {\sideset{^{1}}{^{-1}_{-1}}{\mathop{\tau}}} \&\&\& {\sideset{^{-1}}{^{0}_{0}}{\mathop{\tau}}} \& {\sideset{^{-1}}{^{0}_{1}}{\mathop{\tau}}} \& {\sideset{^{-1}}{^{0}_{2}}{\mathop{\tau}}} \&\&\& {\sideset{^{-3}}{^{1}_{3}}{\mathop{\tau}}} \\
	{\sideset{^{0}}{^{-1}_{-2}}{\mathop{\tau}}} \& {\sideset{^{0}}{^{-1}_{-1}}{\mathop{\tau}}} \& {\sideset{^{0}}{^{-1}_{0}}{\mathop{\tau}}} \&\&\& {\sideset{^{-2}}{^{0}_{1}}{\mathop{\tau}}} \& {\sideset{^{-2}}{^{0}_{2}}{\mathop{\tau}}} \& {\sideset{^{-2}}{^{0}_{3}}{\mathop{\tau}}} \\
	\& {\sideset{^{-1}}{^{-1}_{-1}}{\mathop{\tau}}} \& {\sideset{^{-1}}{^{-1}_{0}}{\mathop{\tau}}} \& {\sideset{^{-1}}{^{-1}_{1}}{\mathop{\tau}}} \&\&\& {\sideset{^{-3}}{^{0}_{2}}{\mathop{\tau}}} \& {\sideset{^{-3}}{^{0}_{3}}{\mathop{\tau}}} \\
	\&\& {\sideset{^{-2}}{^{-1}_{0}}{\mathop{\tau}}} \& {\sideset{^{-2}}{^{-1}_{1}}{\mathop{\tau}}} \& {\sideset{^{-2}}{^{-1}_{2}}{\mathop{\tau}}}
	\arrow[from=1-7, to=1-6]
	\arrow[from=1-8, to=1-7]
	\arrow[from=2-8, to=1-8]
	\arrow[from=2-9, to=2-8]
	\arrow[from=2-7, to=1-7]
	\arrow[from=2-8, to=2-7]
	\arrow[from=3-9, to=2-9]
	\arrow[from=3-8, to=2-8]
	\arrow[from=3-9, to=3-8]
	\arrow[from=1-7, to=2-8]
	\arrow[from=2-8, to=3-9]
	\arrow[from=4-6, to=3-8]
	\arrow[from=2-7, to=4-6]
	\arrow[from=3-5, to=2-7]
	\arrow[from=1-6, to=3-5]
	\arrow[from=2-4, to=1-6]
	\arrow[from=3-4, to=2-4]
	\arrow[from=3-5, to=3-4]
	\arrow[from=4-6, to=4-5]
	\arrow[from=4-5, to=3-5]
	\arrow[from=4-4, to=3-4]
	\arrow[from=4-5, to=4-4]
	\arrow[from=5-6, to=4-6]
	\arrow[from=5-6, to=5-5]
	\arrow[from=5-5, to=4-5]
	\arrow[from=5-7, to=5-6]
	\arrow[from=3-8, to=5-7]
	\arrow[from=3-4, to=4-5]
	\arrow[from=4-5, to=5-6]
	\arrow[from=4-4, to=6-3]
	\arrow[from=6-3, to=5-5]
	\arrow[from=5-5, to=7-4]
	\arrow[from=3-4, to=3-3]
	\arrow[from=3-3, to=2-3]
	\arrow[from=2-4, to=2-3]
	\arrow[from=3-3, to=5-2]
	\arrow[from=5-2, to=4-4]
	\arrow[from=6-2, to=5-2]
	\arrow[from=6-3, to=6-2]
	\arrow[from=7-3, to=6-3]
	\arrow[from=7-3, to=7-2]
	\arrow[from=7-2, to=6-2]
	\arrow[from=7-4, to=7-3]
	\arrow[from=4-9, to=3-9]
	\arrow[from=4-9, to=6-8]
	\arrow[from=5-7, to=4-9]
	\arrow[from=6-7, to=5-7]
	\arrow[from=6-7, to=6-6]
	\arrow[from=6-6, to=5-6]
	\arrow[from=6-8, to=6-7]
	\arrow[from=7-7, to=6-7]
	\arrow[from=7-8, to=7-7]
	\arrow[from=7-8, to=6-8]
	\arrow[from=7-4, to=6-6]
	\arrow[from=8-4, to=7-4]
	\arrow[from=8-5, to=8-4]
	\arrow[from=6-6, to=8-5]
	\arrow[from=8-5, to=7-7]
	\arrow[from=8-3, to=7-3]
	\arrow[from=8-4, to=8-3]
	\arrow[from=6-2, to=7-3]
	\arrow[from=7-3, to=8-4]
	\arrow[from=2-3, to=3-4]
	\arrow[from=5-6, to=6-7]
	\arrow[from=6-7, to=7-8]
	\arrow[from=5-2, to=5-1]
	\arrow[from=6-2, to=6-1]
	\arrow[from=6-1, to=5-1]
	\arrow[from=5-1, to=6-2]
	\arrow[from=3-10, to=3-9]
	\arrow[from=4-10, to=3-10]
	\arrow[from=4-10, to=4-9]
	\arrow[from=3-9, to=4-10]
	\arrow[from=5-1, to=4-1]
	\arrow[from=4-1, to=3-3]
	\arrow[from=5-10, to=4-10]
	\arrow[from=6-8, to=5-10]
\end{tikzcd}}
\end{equation}

By imposing the condition $\sideset{^{n}}{^{m}_{l}}{\mathop{\tau}}=\sideset{^{n}}{^{m+1}_{l+1}}{\mathop{\tau}}, \sigma_{l}^{n}:=\sideset{^{n}}{^{0}_{l}}{\mathop{\tau}}$, Hirota-Miwa equation can be reduced to discrete KdV equation. Meanwhile, we can reduce the quiver (\ref{quiver:hirota-miwa1}) in $(l,m,n)=(1,1,0)$ direction to obtain quiver (\ref{quiver:dkdv1}).

It is interesting that the quiver corresponding to a difference equation is not unique. We give another Hirota-Miwa quiver:

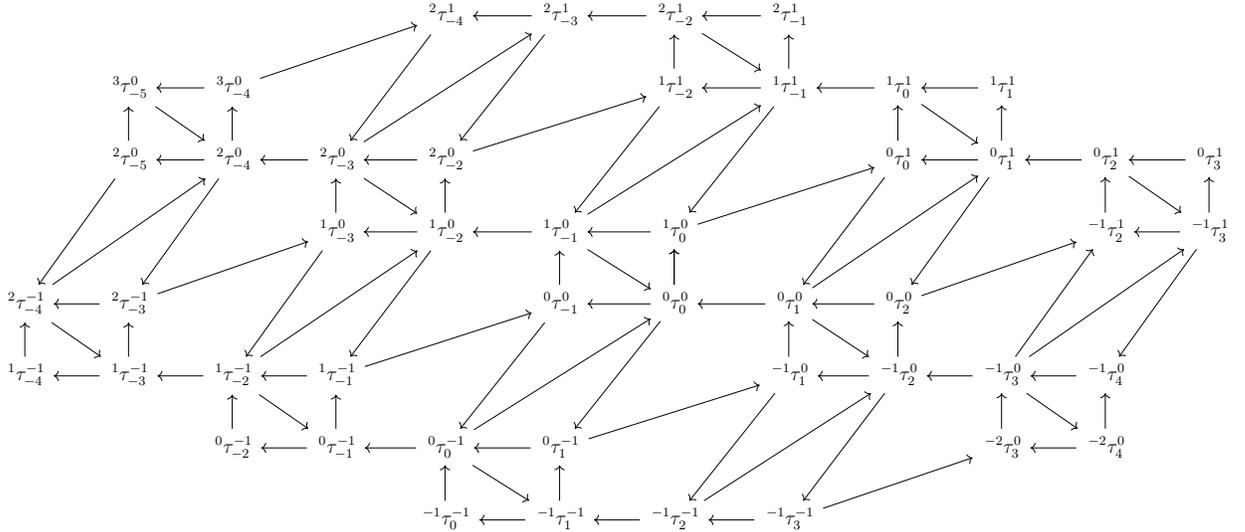
\resizebox{\textwidth}{!}{
\begin{tikzcd}[ampersand replacement=\&]
    \&\&\&\& {\sideset{^{2}}{^{1}_{-4}}{\mathop{\tau}}} \& {\sideset{^{2}}{^{1}_{-3}}{\mathop{\tau}}} \& {\sideset{^{2}}{^{1}_{-2}}{\mathop{\tau}}} \& {\sideset{^{2}}{^{1}_{-1}}{\mathop{\tau}}} \\
	\& {\sideset{^{3}}{^{0}_{-5}}{\mathop{\tau}}} \& {\sideset{^{3}}{^{0}_{-4}}{\mathop{\tau}}} \&\&\&\& {\sideset{^{1}}{^{1}_{-2}}{\mathop{\tau}}} \& {\sideset{^{1}}{^{1}_{-1}}{\mathop{\tau}}} \& {\sideset{^{1}}{^{1}_{0}}{\mathop{\tau}}} \& {\sideset{^{1}}{^{1}_{1}}{\mathop{\tau}}} \\
	\& {\sideset{^{2}}{^{0}_{-5}}{\mathop{\tau}}} \& {\sideset{^{2}}{^{0}_{-4}}{\mathop{\tau}}} \& {\sideset{^{2}}{^{0}_{-3}}{\mathop{\tau}}} \& {\sideset{^{2}}{^{0}_{-2}}{\mathop{\tau}}} \&\&\&\& {\sideset{^{0}}{^{1}_{0}}{\mathop{\tau}}} \& {\sideset{^{0}}{^{1}_{1}}{\mathop{\tau}}} \& {\sideset{^{0}}{^{1}_{2}}{\mathop{\tau}}} \& {\sideset{^{0}}{^{1}_{3}}{\mathop{\tau}}} \\
	\&\&\& {\sideset{^{1}}{^{0}_{-3}}{\mathop{\tau}}} \& {\sideset{^{1}}{^{0}_{-2}}{\mathop{\tau}}} \& {\sideset{^{1}}{^{0}_{-1}}{\mathop{\tau}}} \& {\sideset{^{1}}{^{0}_{0}}{\mathop{\tau}}} \&\&\&\& {\sideset{^{-1}}{^{1}_{2}}{\mathop{\tau}}} \& {\sideset{^{-1}}{^{1}_{3}}{\mathop{\tau}}} \\
	{\sideset{^{2}}{^{-1}_{-4}}{\mathop{\tau}}} \& {\sideset{^{2}}{^{-1}_{-3}}{\mathop{\tau}}} \&\&\&\& {\sideset{^{0}}{^{0}_{-1}}{\mathop{\tau}}} \& {\sideset{^{0}}{^{0}_{0}}{\mathop{\tau}}} \& {\sideset{^{0}}{^{0}_{1}}{\mathop{\tau}}} \& {\sideset{^{0}}{^{0}_{2}}{\mathop{\tau}}} \\
	{\sideset{^{1}}{^{-1}_{-4}}{\mathop{\tau}}} \& {\sideset{^{1}}{^{-1}_{-3}}{\mathop{\tau}}} \& {\sideset{^{1}}{^{-1}_{-2}}{\mathop{\tau}}} \& {\sideset{^{1}}{^{-1}_{-1}}{\mathop{\tau}}} \&\&\&\& {\sideset{^{-1}}{^{0}_{1}}{\mathop{\tau}}} \& {\sideset{^{-1}}{^{0}_{2}}{\mathop{\tau}}} \& {\sideset{^{-1}}{^{0}_{3}}{\mathop{\tau}}} \& {\sideset{^{-1}}{^{0}_{4}}{\mathop{\tau}}} \\
	\&\& {\sideset{^{0}}{^{-1}_{-2}}{\mathop{\tau}}} \& {\sideset{^{0}}{^{-1}_{-1}}{\mathop{\tau}}} \& {\sideset{^{0}}{^{-1}_{0}}{\mathop{\tau}}} \& {\sideset{^{0}}{^{-1}_{1}}{\mathop{\tau}}} \&\&\&\& {\sideset{^{-2}}{^{0}_{3}}{\mathop{\tau}}} \& {\sideset{^{-2}}{^{0}_{4}}{\mathop{\tau}}} \\
	\&\&\&\& {\sideset{^{-1}}{^{-1}_{0}}{\mathop{\tau}}} \& {\sideset{^{-1}}{^{-1}_{1}}{\mathop{\tau}}} \& {\sideset{^{-1}}{^{-1}_{2}}{\mathop{\tau}}} \& {\sideset{^{-1}}{^{-1}_{3}}{\mathop{\tau}}}
	\arrow[from=2-9, to=2-8]
	\arrow[from=2-10, to=2-9]
	\arrow[from=3-10, to=2-10]
	\arrow[from=3-11, to=3-10]
	\arrow[from=3-9, to=2-9]
	\arrow[from=3-10, to=3-9]
	\arrow[from=2-9, to=3-10]
	\arrow[from=3-9, to=5-8]
	\arrow[from=4-7, to=3-9]
	\arrow[from=2-8, to=4-7]
	\arrow[from=4-7, to=4-6]
	\arrow[from=5-8, to=5-7]
	\arrow[from=5-7, to=4-7]
	\arrow[from=5-6, to=4-6]
	\arrow[from=5-7, to=5-6]
	\arrow[from=4-6, to=5-7]
	\arrow[from=4-6, to=4-5]
	\arrow[from=2-8, to=1-8]
	\arrow[from=1-8, to=1-7]
	\arrow[from=4-5, to=3-5]
	\arrow[from=4-11, to=3-11]
	\arrow[from=4-5, to=4-4]
	\arrow[from=3-5, to=3-4]
	\arrow[from=4-4, to=3-4]
	\arrow[from=4-4, to=6-3]
	\arrow[from=6-3, to=6-2]
	\arrow[from=6-2, to=6-1]
	\arrow[from=6-2, to=5-2]
	\arrow[from=5-2, to=5-1]
	\arrow[from=6-1, to=5-1]
	\arrow[from=5-2, to=4-4]
	\arrow[from=3-4, to=3-3]
	\arrow[from=3-3, to=2-3]
	\arrow[from=2-3, to=2-2]
	\arrow[from=3-3, to=5-2]
	\arrow[from=3-4, to=4-5]
	\arrow[from=5-1, to=6-2]
	\arrow[from=5-7, to=4-7]
	\arrow[from=6-9, to=6-8]
	\arrow[from=6-8, to=5-8]
	\arrow[from=6-9, to=5-9]
	\arrow[from=3-3, to=3-2]
	\arrow[from=3-2, to=2-2]
	\arrow[from=6-4, to=6-3]
	\arrow[from=7-4, to=6-4]
	\arrow[from=7-3, to=6-3]
	\arrow[from=7-4, to=7-3]
	\arrow[from=7-5, to=7-4]
	\arrow[from=8-5, to=7-5]
	\arrow[from=7-6, to=7-5]
	\arrow[from=8-6, to=7-6]
	\arrow[from=8-6, to=8-5]
	\arrow[from=8-7, to=8-6]
	\arrow[from=8-8, to=8-7]
	\arrow[from=3-2, to=5-1]
	\arrow[from=5-1, to=3-3]
	\arrow[from=6-3, to=4-5]
	\arrow[from=4-5, to=6-4]
	\arrow[from=6-4, to=5-6]
	\arrow[from=5-6, to=7-5]
	\arrow[from=7-5, to=5-7]
	\arrow[from=5-7, to=7-6]
	\arrow[from=7-6, to=6-8]
	\arrow[from=6-8, to=8-7]
	\arrow[from=2-3, to=1-5]
	\arrow[from=1-5, to=3-4]
	\arrow[from=3-4, to=1-6]
	\arrow[from=1-6, to=3-5]
	\arrow[from=1-6, to=1-5]
	\arrow[from=1-7, to=1-6]
	\arrow[from=4-6, to=2-8]
	\arrow[from=2-7, to=4-6]
	\arrow[from=3-5, to=2-7]
	\arrow[from=2-7, to=1-7]
	\arrow[from=2-8, to=2-7]
	\arrow[from=6-3, to=7-4]
	\arrow[from=7-5, to=8-6]
	\arrow[from=5-9, to=5-8]
	\arrow[from=5-8, to=3-10]
	\arrow[from=3-10, to=5-9]
	\arrow[from=5-9, to=4-11]
	\arrow[from=6-10, to=6-9]
	\arrow[from=7-10, to=6-10]
	\arrow[from=6-11, to=6-10]
	\arrow[from=6-10, to=4-11]
	\arrow[from=7-11, to=7-10]
	\arrow[from=7-11, to=6-11]
	\arrow[from=8-7, to=6-9]
	\arrow[from=6-9, to=8-8]
	\arrow[from=8-8, to=7-10]
	\arrow[from=4-12, to=4-11]
	\arrow[from=4-12, to=3-12]
	\arrow[from=3-12, to=3-11]
	\arrow[from=6-10, to=4-12]
	\arrow[from=4-12, to=6-11]
	\arrow[from=2-2, to=3-3]
	\arrow[from=5-8, to=6-9]
	\arrow[from=6-10, to=7-11]
	\arrow[from=1-7, to=2-8]
	\arrow[from=3-11, to=4-12]
\end{tikzcd}}

 This quiver can also be reduced in $(l,m,n)=(1,1,0)$ direction, then we obtain another quiver corresponding to (\ref{equ:kdv}), which can be further reduced to the somos-4 quiver (\ref{qui:somos4}):

\begin{equation}\label{quiver:dkdv2}
    \begin{tikzcd}
	{\sigma_{-2}^{2}} & {\sigma_{-1}^{2}} \\
	{\sigma_{-2}^{1}} & {\sigma_{-1}^{1}} & {\sigma_{0}^{1}} & {\sigma_{1}^{1}} \\
	&& {\sigma_{0}^{0}} & {\sigma_{1}^{0}} & {\sigma_{2}^{0}} & {\sigma_{3}^{0}} \\
	&&&& {\sigma_{2}^{-1}} & {\sigma_{3}^{-1}} & {\sigma_{4}^{-1}}
	\arrow[from=2-4, to=2-3]
	\arrow[from=3-4, to=2-4]
	\arrow[from=4-5, to=3-5]
	\arrow[from=4-6, to=4-5]
	\arrow[from=3-6, to=3-5]
	\arrow["2"{description}, from=3-5, to=3-4]
	\arrow[from=4-6, to=3-6]
	\arrow[from=3-5, to=4-6]
	\arrow[from=3-6, to=4-5]
	\arrow["2"{description}, from=4-7, to=4-6]
	\arrow[curve={height=12pt}, from=4-5, to=4-7]
	\arrow[curve={height=12pt}, from=3-4, to=3-6]
	\arrow[from=3-4, to=3-3]
	\arrow[from=3-3, to=2-3]
	\arrow[from=2-4, to=3-3]
	\arrow[from=2-3, to=3-4]
	\arrow[curve={height=12pt}, from=3-3, to=3-5]
	\arrow["2"{description}, from=2-3, to=2-2]
	\arrow[from=2-2, to=1-2]
	\arrow[from=2-2, to=2-1]
	\arrow[from=1-2, to=1-1]
	\arrow[from=2-1, to=1-1]
	\arrow[from=1-2, to=2-1]
	\arrow[from=1-1, to=2-2]
	\arrow[curve={height=12pt}, from=2-1, to=2-3]
	\arrow[curve={height=12pt}, from=2-2, to=2-4]
\end{tikzcd}
\end{equation}

\subsection{The Poisson structure of the discrete KdV equation}

We need to recall some results in \cite{inoue2011difference} first, then we will give an assumption and calculate the Poisson structure of the discrete KdV equation.

Now we forget the $\sigma_{l}^{n}$ temporarily, and only pay attention to the structure of (\ref{quiver:dkdv1}). Let's mark the label of vertices again:

\begin{equation}\label{quiver:number}
    \begin{tikzcd}
	x_{-3} & x_{-2} & x_{-1} \\
	& x_{0} & x_{1} & x_{2} \\
	&& x_{3} & x_{4} & x_{5} \\
	&&& x_{6} & x_{7} & x_{8}
	\arrow[from=2-2, to=1-2]
	\arrow[from=2-3, to=2-2]
	\arrow[from=3-3, to=2-3]
	\arrow[from=4-4, to=3-4]
	\arrow[from=1-2, to=1-1]
	\arrow[from=4-5, to=4-4]
	\arrow[from=2-4, to=2-3]
	\arrow[curve={height=12pt}, from=2-2, to=2-4]
	\arrow[from=2-4, to=3-3]
	\arrow[from=3-4, to=2-4]
	\arrow[from=2-3, to=3-4]
	\arrow[from=1-3, to=1-2]
	\arrow[curve={height=12pt}, from=1-1, to=1-3]
	\arrow[from=2-3, to=1-3]
	\arrow[from=1-2, to=2-3]
	\arrow[from=1-3, to=2-2]
	\arrow[from=3-5, to=3-4]
	\arrow[from=3-4, to=3-3]
	\arrow[from=4-5, to=3-5]
	\arrow[from=3-4, to=4-5]
	\arrow[from=3-5, to=4-4]
	\arrow[from=4-6, to=4-5]
	\arrow[curve={height=12pt}, from=4-4, to=4-6]
	\arrow[curve={height=12pt}, from=3-3, to=3-5]
\end{tikzcd}
\end{equation}
and the translation is 
\begin{equation}\label{equ:trans}
    x_{i}\mapsto \sigma_{l}^{n}=\begin{cases}\sigma_{k}^{-k} & i=3k, \\ 
    \sigma_{k+1}^{-k} & i=3k +1, \quad\quad\quad \text{for  } k\in \mathbb{Z}.\\
    \sigma_{k+2}^{-k} & i=3k+2,\end{cases}
\end{equation}

the infinite dimensional skew-symmetric matrix $B=(b_{ij})_{i,j\in \mathbb{Z}}$ corresponding to (\ref{quiver:number}) is written as:

\begin{equation*}
\begin{aligned}
& b_{3 k, 3 k+i}=\delta_{i,-2}-\delta_{i,-1}-\delta_{i, 1}+\delta_{i, 2},\\
& b_{3 k+1,3 k+1+i}=-\delta_{i,-3}+\delta_{i,-2}+\delta_{i,-1}-\delta_{i, 1}-\delta_{i, 2}+\delta_{i, 3},\quad\quad i,k \in \mathbb{Z}.\\
& b_{3 k+2,3 k+2+i}=-\delta_{i,-2}+\delta_{i,-1}+\delta_{i, 1}-\delta_{i, 2},
\end{aligned}
\end{equation*}
Since $b_{j,3k}+b_{j,3k+1}+b_{j,3k+2}=0 $, $B$ has no left inverse. By the theorem \ref{thm:pb0}, we need to compute the infinite dimensional skew-symmetric matrix $P=(p_{ij})_{i,j\in \mathbb{Z}}$ satisfying $PB=O$. Define 3 by 3 submatrices of $P$: 
\begin{equation*}
    P(i,j)=\left(\begin{array}{ccc}
    p_{3i,3j} & p_{3i,3j +1} & p_{3i,3j+2}  \\
    p_{3i+1,3j} & p_{3i+1,3j+1} & p_{3i+1,3j+2}  \\
    p_{3i+2,3j} & p_{3i+2,3j+1} & p_{3i+2,3j+2}  
    \end{array}\right),\quad i,j \in \mathbb{Z}.
\end{equation*}

\begin{thm}\cite{inoue2011difference}\label{thm:inoue}
The general skew-symmetric matrix solution $P$ to $PB=O$ is given by:
\begin{equation}\label{equ:inoue}
\begin{array}{lll}
P(0,0) & =\left(\begin{array}{ccc}
0 & a_0 & b_0 \\
-a_0 & 0 & c_0 \\
-b_0 & -c_0 & 0
\end{array}\right) & a_0, b_0, c_0 \in \mathbb{Q}, \\
P(i, i) & =P(0,0)+Q_{0, i} S-S Q_{0, i} \quad\quad& i \neq 0, \\
P(i, j) & =P(i,i)+Q_{i, j} S & i<j, \\
P(j, i) & =-P(i, j)^T & i<j .
\end{array}
\end{equation}
where $S=\left(\begin{array}{ccc}
     1 & 1 & 1 \\
     1 & 1 & 1 \\
     1 & 1 & 1 
\end{array}\right)$ and $Q_{i,j}=\left(\begin{array}{ccc}
     q_{i,j} & 0 & 0 \\
     0 & q_{i,j}+a(i)-a(j) & 0 \\
     0 & 0 & q_{i,j}+b(i)-b(j)
\end{array}\right)$for two fixed maps $a,b : \mathbb{Z} \rightarrow \mathbb{Q}$ and $q_{i,j}=-q_{j,i}$.
\end{thm}

The above Poisson structure still has a high degree of freedom, we will add more restrictions to it and get the general form of the Hamiltonian structure. 

One thing worth considering is that the Poisson structure defined on cluster algebra follows the quiver, it can only describe the Poisson bracket between variables in the same quiver. 
If we want to study the Poisson structure of the difference equation, we should allow the the bracket between variables which are in different quivers. Since the Poisson bracket we considered is log-canonical, we assume the Poisson structure of the difference equation is also log-canonical.

For $i,j \in \mathbb{Z}$, consider 
\begin{equation}\label{equ:coef}
    \begin{split}
        &\{x_i,\frac{x_{3j+1}x_{3j-1}+x_{3j+2}x_{3j-2}}{x_{3j}}\}\\
        =& (p_{i,3j+1}+p_{i,3j-1}-p_{i,3j})x_i \frac{x_{3j+1}x_{3j-1}}{x_{3j}}+(p_{i,3j+2}+p_{i,3j-2}-p_{i,3j})x_i \frac{x_{3j+2}x_{3j-2}}{x_{3j}}\\
        =& (\text{some coefficient})x_i\frac{x_{3j+1}x_{3j-1}+x_{3j+2}x_{3j-2}}{x_{3j}},
    \end{split}
\end{equation}
that means 
\begin{equation}\label{equ:pb01}
    p_{i,3j+1}+p_{i,3j-1}=p_{i,3j+2}+p_{i,3j-2}.
\end{equation}
Similarly, by considering the composition of mutation in the other direction, we get 
\begin{equation}\label{equ:pb02}
    p_{i,3j+1}+p_{i,3j+3}=p_{i,3j}+p_{i,3j+4}.
\end{equation}
In fact, equation (\ref{equ:pb01}) and (\ref{equ:pb02}) are equivalent to $PB=O$, which is also equivalent to $\{J_Z, \mathcal{R}\}\subseteq J_Z$ in definition \ref{definition}, so we can extend the log-canonical Poisson bracket to any variables in equation (\ref{equ:kdv}). Hence the Poisson bracket of a pair of variables may not only be calculated by extending the initial Poisson bracket, but also be given by the Poisson structure of some quiver. They are compatible. 

By the condition (\ref{traninvariant}), we require
\begin{equation}\label{equ:assumption}
    \{\Lambda_{l}\sigma_{l}^{n},\Lambda_{l}\sigma_{l'}^{n'}\}=\Lambda_{l}\{\sigma_{l}^{n},\sigma_{l'}^{n'}\}.
\end{equation}

Consider the Poisson structures satisfying the following relations:

\begin{equation}\label{equ:relation}
    \begin{array}{ll}
        p_{i,3j+1}+p_{i,3j-1}-p_{i,3j}=p_{i-1,3j-1}\quad\quad&i \equiv 1,2 \quad\text{mod } 3   \\
        p_{i,3j+1}+p_{i,3j+3}-p_{i,3j+2}=p_{i+1,3j+3} &  i \equiv 0,1 \quad\text{mod } 3   \\
        p_{i,j}=p_{i-1,j-1} & i,j \not\equiv 0 \quad\text{mod } 3 
    \end{array}
\end{equation}
Define function
\begin{equation*}
	F(\{\sigma_{l}^{n},\sigma_{l'}^{n'}\})=\frac{\{\sigma_{l}^{n},\sigma_{l'}^{n'}\}}{\sigma_{l}^{n}\sigma_{l'}^{n'}},
\end{equation*}
which means the coefficients of log-canonical Poisson bracket.

The first relation in (\ref{equ:relation}) comes from (\ref{equ:coef}) and that 
\begin{equation*}
	\begin{split}
		&F(\{x_i,\frac{x_{3j+1}x_{3j-1}+x_{3j+2}x_{3j-2}}{x_{3j}}\})\\
		\overset{ (\ref{equ:trans})}{=}&F(\{\sigma_{l}^{n},\Tilde{\mu_0}(\sigma_{k}^{-k})\})=F(\{\sigma_{l}^{n},\sigma_{k+2}^{-k+1}\}) \\
		\overset{ (\ref{equ:assumption})}{=}&F(\{\sigma_{l-1}^{n},\sigma_{k+1}^{-k+1}\})\overset{(\ref{equ:trans})}{=}F(\{x_{i-1}, x_{3j-1}\}).
	\end{split}
\end{equation*}
Similarly, the second relation comes from the composition of mutation in the other direction. The third relation is (\ref{equ:assumption}) in the case that there are two quivers in which $\sigma_{l}^{n},\sigma_{l'}^{n'}$ both exist.

\begin{thm}
    A Poisson bracket is a Hamiltonian structure of the discrete KdV equation (\ref{equ:kdv}) if and only if the Poisson bracket is given by:
\begin{equation}\label{equ:poisson}
    \{\sigma_{l}^{n},\sigma_{l'}^{n'}\}=
        (q_{n-n'}+a_0(n'+l'-n-l))\sigma_{l}^{n}\sigma_{l'}^{n'},
\end{equation}
where $a_0 \in \mathbb{Q}$ and $q_{k} \in \mathbb{Q}$ satisfying $q_{k}=-q_{-k}$. This Poisson bracket gives a bihamiltonian structure of the discrete KdV equation.
\end{thm}
\begin{proof}
	It is calculated directly that (\ref{equ:poisson}) satisfies the definition(\ref{definition}).

    On the other hand, by the theorem \ref{thm:inoue}, $P$ has the form (\ref{equ:inoue}).
    For the third relation in (\ref{equ:relation}), consider the following three cases in turn:
    \begin{itemize}
        \item If $i=1, j=2$, then $p_{1,2}=p_{0,1}, i.e.,a_0=c_0$. 
        \item If $i=1, j=3k+2 ,k \neq 0$, then $a_0=a_0+a(0)-a(k),i.e.,a$ is a constant map.
        \item If $i=2, j=3k+1 ,k \neq 0$, then $a_0=a_0+b(0)-b(k),i.e.,b$ is a constant map.
    \end{itemize}
    Hence $Q_{i,j}$ in (\ref{equ:inoue}) is $q_{i,j}id$. By the translation invariant (\ref{traninvariant}), we know $Q_{i,j}$ should be $q_{i-j}id$.

    Consider the first relation in (\ref{equ:relation}) when $i=1, j=0$. We have $0+ q_{0,-1}+a_0 -(-a_0)= q_{0,-1}+b_0$, hence $b_0 = 2 a_0$.
    
    By translation (\ref{equ:trans}), we get that (\ref{equ:poisson}) holds in the case $0\leq n+l,n'+l' \leq 2$. Assume that (\ref{equ:poisson}) holds in the case $0\leq n+l,n'+l' \leq k$ and use induction on $k$. 
    
    Consider $\{\sigma_{l_1}^{n_1},\sigma_{l_2}^{n_2}\}, 0\leq n_1+l_1,n_2+l_2 \leq k+1$.
    \begin{itemize}
        \item If both $ n_1+l_1,n_2+l_2 \leq k$, by assumption,  (\ref{equ:poisson}) holds.
        \item If $0 \leq n_1+l_1 \leq k,n_2+l_2=k+1$, 
        \begin{equation}\label{equ:induction}
        \begin{split}
            &\{\sigma_{l_1}^{n_1},\sigma_{l_2}^{n_2} \}= \{\sigma_{l_1}^{n_1},\frac{\sigma_{l_2-1}^{n_2}\sigma_{l_2-1}^{n_2-1}+\sigma_{l_2-2}^{n_2}\sigma_{l_2}^{n_2-1}}{\sigma_{l_2-2}^{n_2-1}}\}\\
            =&(q_{n_1-n_2}+a_0(n_2+l_2-1-n_1-l_1)+q_{n_1- n_2+1}+a_0(n_2-1+l_2-1-n_1-l_1)\\
            &-q_{n_1- n_2+1}-a_0(n_2-1+l_2-2-n_1-l_1))\sigma_{l_1}^{n_1}\frac{\sigma_{l_2-1}^{n_2}\sigma_{l_2-1}^{n_2-1}+\sigma_{l_2-2}^{n_2}\sigma_{l_2}^{n_2-1}}{\sigma_{l_2-2}^{n_2-1}}\\
            =&(q_{n_1-n_2}+a_0(n_2 +l_2 -n_1-l_1))\sigma_{l_1}^{n_1}\sigma_{l_2}^{n_2}
        \end{split}
        \end{equation}
        \item If both $ n_1+l_1,n_2+l_2=k+1$, the proof is similar to (\ref{equ:induction}). Specially we need to use the results of (\ref{equ:induction}).
    \end{itemize}
    By induction, we get that (\ref{equ:poisson}) holds in the case $ n+l,n'+l' \geq 0$. The proof for the negative direction of the real axis is similar.
    
\end{proof}
\renewcommand\qedsymbol{QED}

\begin{remark}
	The Poisson bracket (\ref{equ:poisson}) was also obtained by Inoue and  Nakanishi in \cite{inoue2011difference} by using a certain symmetry condition on the Poisson bracket \eqref{equ:inoue}.
\end{remark}

\begin{remark}
	The Poisson bracket (\ref{equ:poisson}) can be simplified:
	\begin{equation*}
		\{\sigma_{l}^{n},\sigma_{l'}^{n'}\}=
		(q_{n-n'}+a_0(l'-l))\sigma_{l}^{n}\sigma_{l'}^{n'},
	\end{equation*}
	where $a_0 \in \mathbb{Q}$ and $q_{k} \in \mathbb{Q}$ satisfying $q_{k}=-q_{-k}$.
\end{remark}

Note that we choose the quiver (\ref{quiver:dkdv1}) corresponding to the discrete KdV equation then derive the Hamiltonian structure (\ref{equ:poisson}). If we choose the quiver (\ref{quiver:dkdv2}), we can also get (\ref{equ:poisson}). In fact, if a quiver can be obtained from the (\ref{quiver:dkdv1}) with a multiplicity of mutations, according to compatibility the Hamiltonian structure calculated from this quiver is still (\ref{equ:poisson}), that is, the form of the Hamiltonian structure is independent of the quiver.

\section{Other reductions of Hirota-Miwa equation}
In subsection 4.1, we know that the quiver of the discrete KdV equation can be obtained from the quiver of Hirota-Miwa equation. We will write some principles for constructing quivers of (\ref{equ:hirota-miwa}) and give some interesting examples.

The first principle is the quiver of (\ref{equ:hirota-miwa}) should contain the following two structures:
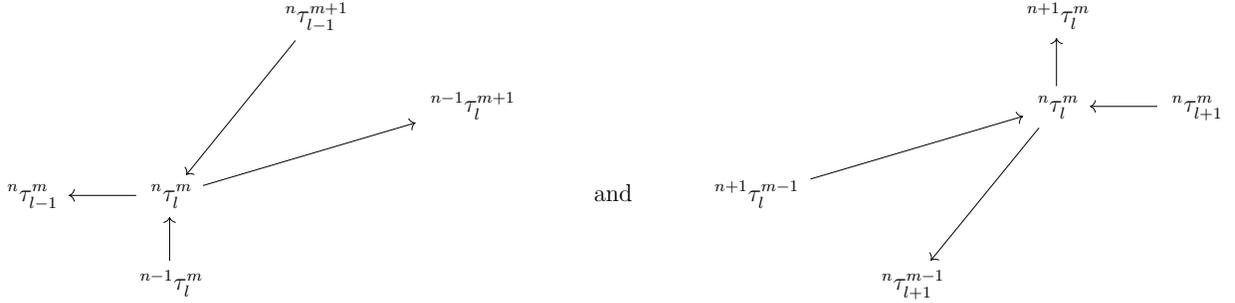
\begin{equation*}
    \resizebox{\textwidth}{!}{
    \begin{tikzcd}[ampersand replacement=\&]
	\&\& {\sideset{^{n}}{^{m+1}_{l-1}}{\mathop{\tau}}} \&\&\&\&\& {\sideset{^{n+1}}{^{m}_{l}}{\mathop{\tau}}} \\
	\&\&\& {\sideset{^{n-1}}{^{m+1}_{l}}{\mathop{\tau}}} \&\&\&\& {\sideset{^{n}}{^{m}_{l}}{\mathop{\tau}}} \& {\sideset{^{n}}{^{m}_{l+1}}{\mathop{\tau}}} \\
	{\sideset{^{n}}{^{m}_{l-1}}{\mathop{\tau}}} \& {\sideset{^{n}}{^{m}_{l}}{\mathop{\tau}}} \&\&\& \text{and} \& {\sideset{^{n+1}}{^{m-1}_{l}}{\mathop{\tau}}} \\
	\& {\sideset{^{n-1}}{^{m}_{l}}{\mathop{\tau}}} \&\&\&\&\& {\sideset{^{n}}{^{m-1}_{l+1}}{\mathop{\tau}}}
	\arrow[from=3-2, to=3-1]
	\arrow[from=4-2, to=3-2]
	\arrow[from=1-3, to=3-2]
	\arrow[from=3-2, to=2-4]
	\arrow[from=2-9, to=2-8]
	\arrow[from=2-8, to=1-8]
	\arrow[from=2-8, to=4-7]
	\arrow[from=3-6, to=2-8]
\end{tikzcd}}.
\end{equation*}
This principle guarantees the variables mutation formula is equation (\ref{equ:hirota-miwa}).

The second principle is periodic, in other word, the quiver has the same structure after the given composition of variable mutations. And we also request that each variable $\sideset{^{n}}{^{m}_{l}}{\mathop{\tau}}$ will appear at the position where it should be mutated in some quiver. 

We give two other reductions of Hirota-Miwa equation.
\begin{eg}
Given the following quiver of equation (\ref{equ:hirota-miwa}) and the composition of mutations at $\{\sideset{^{n}}{^{m}_{l}}{\mathop{\tau}}|n-2m+l=-2\}$.

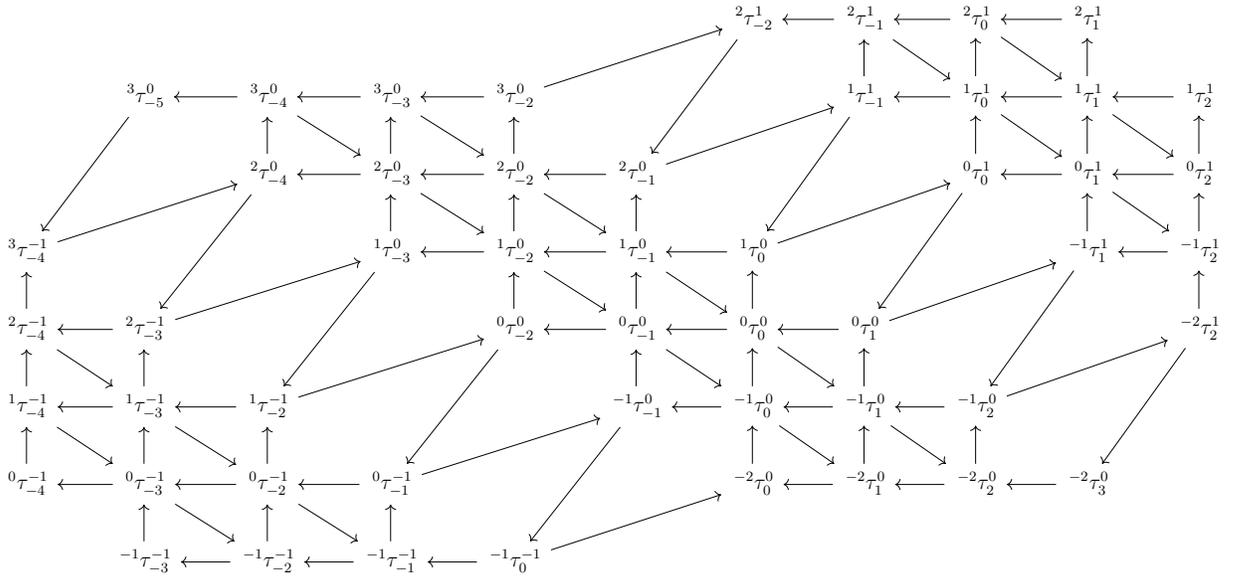
\resizebox{\textwidth}{!}{
    \begin{tikzcd}[ampersand replacement=\&]
	\&\&\&\&\&\& {\sideset{^{2}}{^{1}_{-2}}{\mathop{\tau}}} \& {\sideset{^{2}}{^{1}_{-1}}{\mathop{\tau}}} \& {\sideset{^{2}}{^{1}_{0}}{\mathop{\tau}}} \& {\sideset{^{2}}{^{1}_{1}}{\mathop{\tau}}} \\
	\& {\sideset{^{3}}{^{0}_{-5}}{\mathop{\tau}}} \& {\sideset{^{3}}{^{0}_{-4}}{\mathop{\tau}}} \& {\sideset{^{3}}{^{0}_{-3}}{\mathop{\tau}}} \& {\sideset{^{3}}{^{0}_{-2}}{\mathop{\tau}}} \&\&\& {\sideset{^{1}}{^{1}_{-1}}{\mathop{\tau}}} \& {\sideset{^{1}}{^{1}_{0}}{\mathop{\tau}}} \& {\sideset{^{1}}{^{1}_{1}}{\mathop{\tau}}} \& {\sideset{^{1}}{^{1}_{2}}{\mathop{\tau}}} \\
	\&\& {\sideset{^{2}}{^{0}_{-4}}{\mathop{\tau}}} \& {\sideset{^{2}}{^{0}_{-3}}{\mathop{\tau}}} \& {\sideset{^{2}}{^{0}_{-2}}{\mathop{\tau}}} \& {\sideset{^{2}}{^{0}_{-1}}{\mathop{\tau}}} \&\&\& {\sideset{^{0}}{^{1}_{0}}{\mathop{\tau}}} \& {\sideset{^{0}}{^{1}_{1}}{\mathop{\tau}}} \& {\sideset{^{0}}{^{1}_{2}}{\mathop{\tau}}} \\
	{\sideset{^{3}}{^{-1}_{-4}}{\mathop{\tau}}} \&\&\& {\sideset{^{1}}{^{0}_{-3}}{\mathop{\tau}}} \& {\sideset{^{1}}{^{0}_{-2}}{\mathop{\tau}}} \& {\sideset{^{1}}{^{0}_{-1}}{\mathop{\tau}}} \& {\sideset{^{1}}{^{0}_{0}}{\mathop{\tau}}} \&\&\& {\sideset{^{-1}}{^{1}_{1}}{\mathop{\tau}}} \& {\sideset{^{-1}}{^{1}_{2}}{\mathop{\tau}}} \\
	{\sideset{^{2}}{^{-1}_{-4}}{\mathop{\tau}}} \& {\sideset{^{2}}{^{-1}_{-3}}{\mathop{\tau}}} \&\&\& {\sideset{^{0}}{^{0}_{-2}}{\mathop{\tau}}} \& {\sideset{^{0}}{^{0}_{-1}}{\mathop{\tau}}} \& {\sideset{^{0}}{^{0}_{0}}{\mathop{\tau}}} \& {\sideset{^{0}}{^{0}_{1}}{\mathop{\tau}}} \&\&\& {\sideset{^{-2}}{^{1}_{2}}{\mathop{\tau}}} \\
	{\sideset{^{1}}{^{-1}_{-4}}{\mathop{\tau}}} \& {\sideset{^{1}}{^{-1}_{-3}}{\mathop{\tau}}} \& {\sideset{^{1}}{^{-1}_{-2}}{\mathop{\tau}}} \&\&\& {\sideset{^{-1}}{^{0}_{-1}}{\mathop{\tau}}} \& {\sideset{^{-1}}{^{0}_{0}}{\mathop{\tau}}} \& {\sideset{^{-1}}{^{0}_{1}}{\mathop{\tau}}} \& {\sideset{^{-1}}{^{0}_{2}}{\mathop{\tau}}} \\
	{\sideset{^{0}}{^{-1}_{-4}}{\mathop{\tau}}} \& {\sideset{^{0}}{^{-1}_{-3}}{\mathop{\tau}}} \& {\sideset{^{0}}{^{-1}_{-2}}{\mathop{\tau}}} \& {\sideset{^{0}}{^{-1}_{-1}}{\mathop{\tau}}} \&\&\& {\sideset{^{-2}}{^{0}_{0}}{\mathop{\tau}}} \& {\sideset{^{-2}}{^{0}_{1}}{\mathop{\tau}}} \& {\sideset{^{-2}}{^{0}_{2}}{\mathop{\tau}}} \& {\sideset{^{-2}}{^{0}_{3}}{\mathop{\tau}}} \\
	\& {\sideset{^{-1}}{^{-1}_{-3}}{\mathop{\tau}}} \& {\sideset{^{-1}}{^{-1}_{-2}}{\mathop{\tau}}} \& {\sideset{^{-1}}{^{-1}_{-1}}{\mathop{\tau}}} \& {\sideset{^{-1}}{^{-1}_{0}}{\mathop{\tau}}}
	\arrow[from=2-9, to=2-8]
	\arrow[from=2-10, to=2-9]
	\arrow[from=3-10, to=2-10]
	\arrow[from=3-11, to=3-10]
	\arrow[from=3-9, to=2-9]
	\arrow[from=3-10, to=3-9]
	\arrow[from=2-9, to=3-10]
	\arrow[from=3-9, to=5-8]
	\arrow[from=4-7, to=3-9]
	\arrow[from=2-8, to=4-7]
	\arrow[from=4-7, to=4-6]
	\arrow[from=5-8, to=5-7]
	\arrow[from=5-7, to=4-7]
	\arrow[from=5-6, to=4-6]
	\arrow[from=5-7, to=5-6]
	\arrow[from=6-8, to=5-8]
	\arrow[from=4-6, to=5-7]
	\arrow[from=4-6, to=4-5]
	\arrow[from=5-6, to=5-5]
	\arrow[from=5-5, to=4-5]
	\arrow[from=4-5, to=5-6]
	\arrow[from=3-11, to=2-11]
	\arrow[from=2-11, to=2-10]
	\arrow[from=2-10, to=3-11]
	\arrow[from=2-8, to=1-8]
	\arrow[from=1-8, to=1-7]
	\arrow[from=6-7, to=5-7]
	\arrow[from=6-7, to=6-6]
	\arrow[from=6-8, to=6-7]
	\arrow[from=6-6, to=5-6]
	\arrow[from=5-6, to=6-7]
	\arrow[from=5-7, to=6-8]
	\arrow[from=2-9, to=1-9]
	\arrow[from=1-9, to=1-8]
	\arrow[from=4-6, to=3-6]
	\arrow[from=3-6, to=3-5]
	\arrow[from=4-5, to=3-5]
	\arrow[from=3-5, to=4-6]
	\arrow[from=3-6, to=2-8]
	\arrow[from=1-7, to=3-6]
	\arrow[from=3-5, to=2-5]
	\arrow[from=2-5, to=1-7]
	\arrow[from=1-10, to=1-9]
	\arrow[from=2-10, to=1-10]
	\arrow[from=6-9, to=6-8]
	\arrow[from=4-11, to=3-11]
	\arrow[from=4-10, to=3-10]
	\arrow[from=4-11, to=4-10]
	\arrow[from=5-8, to=4-10]
	\arrow[from=4-10, to=6-9]
	\arrow[from=1-8, to=2-9]
	\arrow[from=1-9, to=2-10]
	\arrow[from=3-10, to=4-11]
	\arrow[from=4-5, to=4-4]
	\arrow[from=3-5, to=3-4]
	\arrow[from=4-4, to=3-4]
	\arrow[from=2-5, to=2-4]
	\arrow[from=3-4, to=2-4]
	\arrow[from=5-5, to=7-4]
	\arrow[from=7-4, to=6-6]
	\arrow[from=6-6, to=8-5]
	\arrow[from=8-5, to=7-7]
	\arrow[from=7-7, to=6-7]
	\arrow[from=7-8, to=7-7]
	\arrow[from=7-8, to=6-8]
	\arrow[from=7-9, to=6-9]
	\arrow[from=7-9, to=7-8]
	\arrow[from=7-10, to=7-9]
	\arrow[from=5-11, to=4-11]
	\arrow[from=6-9, to=5-11]
	\arrow[from=5-11, to=7-10]
	\arrow[from=6-3, to=5-5]
	\arrow[from=4-4, to=6-3]
	\arrow[from=7-3, to=6-3]
	\arrow[from=6-3, to=6-2]
	\arrow[from=7-3, to=7-2]
	\arrow[from=7-2, to=6-2]
	\arrow[from=8-4, to=7-4]
	\arrow[from=8-4, to=8-3]
	\arrow[from=8-3, to=7-3]
	\arrow[from=8-5, to=8-4]
	\arrow[from=8-3, to=8-2]
	\arrow[from=8-2, to=7-2]
	\arrow[from=7-2, to=7-1]
	\arrow[from=7-1, to=6-1]
	\arrow[from=6-2, to=6-1]
	\arrow[from=6-2, to=5-2]
	\arrow[from=5-2, to=5-1]
	\arrow[from=6-1, to=5-1]
	\arrow[from=5-1, to=4-1]
	\arrow[from=5-2, to=4-4]
	\arrow[from=3-4, to=3-3]
	\arrow[from=3-3, to=2-3]
	\arrow[from=2-4, to=2-3]
	\arrow[from=2-3, to=2-2]
	\arrow[from=4-1, to=3-3]
	\arrow[from=3-3, to=5-2]
	\arrow[from=2-2, to=4-1]
	\arrow[from=2-3, to=3-4]
	\arrow[from=3-4, to=4-5]
	\arrow[from=2-4, to=3-5]
	\arrow[from=6-7, to=7-8]
	\arrow[from=6-8, to=7-9]
	\arrow[from=5-1, to=6-2]
	\arrow[from=6-1, to=7-2]
	\arrow[from=6-2, to=7-3]
	\arrow[from=7-2, to=8-3]
	\arrow[from=7-3, to=8-4]
	\arrow[from=7-4, to=7-3]
\end{tikzcd}
}

By imposing the condition $\sideset{^{n}}{^{m}_{l}}{\mathop{\tau}}=\sideset{^{n+1}}{^{m+1}_{l+1}}{\mathop{\tau}}, \sigma_{l}^{n}:=\sideset{^{n}}{^{1}_{l}}{\mathop{\tau}}$ on (\ref{equ:hirota-miwa}), we will get a new partial difference equation:

\begin{equation}\label{equ:reduction111}
    \sigma_{l+2}^{n+2}=\frac{\sigma_{l+1}^{n} \sigma_{l+1}^{n+2}+\sigma_l^{n+1} \sigma_{l+2}^{n+1}}{\sigma_{l}^n}.
\end{equation}
This equation is the discrete Boussinesq equation\cite{date1983method}\cite{zabrodin1997hirota}.

Reduce the above quiver in $(l,m,n)=(1,1,1)$ direction, we get the corresponding quiver: 

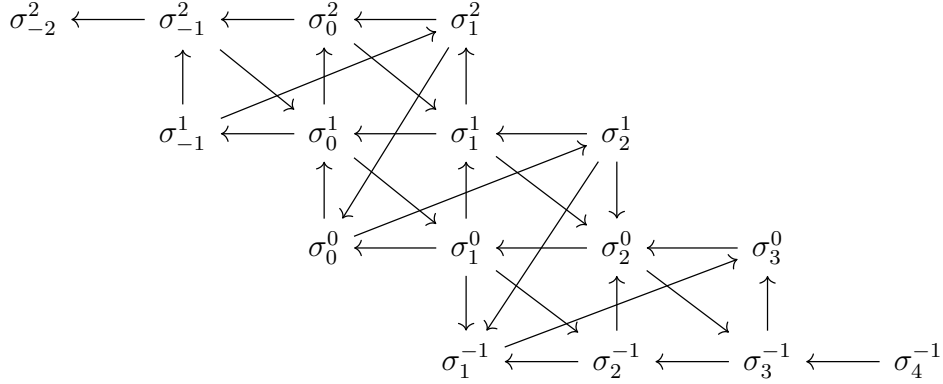
\[\begin{tikzcd}
	{\sigma_{-2}^{2}} & {\sigma_{-1}^{2}} & {\sigma_{0}^{2}} & {\sigma_{1}^{2}} \\
	& {\sigma_{-1}^{1}} & {\sigma_{0}^{1}} & {\sigma_{1}^{1}} & {\sigma_{2}^{1}} \\
	&& {\sigma_{0}^{0}} & {\sigma_{1}^{0}} & {\sigma_{2}^{0}} & {\sigma_{3}^{0}} \\
	&&& {\sigma_{1}^{-1}} & {\sigma_{2}^{-1}} & {\sigma_{3}^{-1}} & {\sigma_{4}^{-1}}
	\arrow[from=2-4, to=2-3]
	\arrow[from=3-4, to=2-4]
	\arrow[from=4-5, to=3-5]
	\arrow[from=4-6, to=4-5]
	\arrow[from=3-6, to=3-5]
	\arrow[from=3-5, to=3-4]
	\arrow[from=4-6, to=3-6]
	\arrow[from=3-5, to=4-6]
	\arrow[from=4-7, to=4-6]
	\arrow[from=3-4, to=3-3]
	\arrow[from=3-3, to=2-3]
	\arrow[from=2-3, to=3-4]
	\arrow[from=2-3, to=2-2]
	\arrow[from=2-2, to=1-2]
	\arrow[from=1-2, to=1-1]
	\arrow[from=2-5, to=2-4]
	\arrow[from=2-5, to=3-5]
	\arrow[from=3-4, to=4-4]
	\arrow[from=4-5, to=4-4]
	\arrow[from=2-4, to=1-4]
	\arrow[from=1-4, to=1-3]
	\arrow[from=1-3, to=1-2]
	\arrow[from=2-3, to=1-3]
	\arrow[from=1-2, to=2-3]
	\arrow[from=1-3, to=2-4]
	\arrow[from=2-4, to=3-5]
	\arrow[from=3-4, to=4-5]
	\arrow[from=1-4, to=3-3]
	\arrow[from=3-3, to=2-5]
	\arrow[from=2-5, to=4-4]
	\arrow[from=4-4, to=3-6]
	\arrow[from=2-2, to=1-4]
\end{tikzcd}\]

In general case, equation $PB=0$ implies  $\{J_Z, \mathcal{R}\}\subseteq J_Z$. 
Hence a Hamiltonian structure $P=(p_{ij})$ of (\ref{equ:reduction111}) satisfies:

\begin{equation*}
    \begin{split}
        &p_{i,4k-5}-p_{i,4k-3}-p_{i,4k-1}+p_{i,4k+1}=0,\\
        &p_{i,4k-3}-p_{i,4k-2}-p_{i,4k}+p_{i,4k+2}+p_{i,4k+4}-p_{i,4k+5}=0,\\
        &p_{i,4k-2}-p_{i,4k-1}-p_{i,4k+1}+p_{i,4k+3}+p_{i,4k+5}-p_{i,4k+6}=0,\\
        &p_{i,4k+2}-p_{i,4k+4}-p_{i,4k+6}+p_{i,4k+8}=0
    \end{split}
\end{equation*}
and
\begin{equation*}
    \begin{array}{ll}
        p_{i,4k-3}+p_{i,4k-1}-p_{i,4k}=p_{i-1,4k-1}\quad\quad&i \equiv 1,2,3 \quad\text{mod } 4,   \\
        p_{i,4k}+p_{i,4k+2}-p_{i,4k-1}=p_{i+1,4k} &  i \equiv 0,1,2 \quad\text{mod } 4,   \\
        p_{i,j}=p_{i-1,j-1} & i,j \not\equiv 0 \quad\text{mod } 4. 
    \end{array}
\end{equation*}
and
\begin{equation*}
	\{\Lambda_{n}\sigma_{l}^{n},\Lambda_{n}\sigma_{l'}^{n'}\}=\Lambda_{n}\{\sigma_{l}^{n},\sigma_{l'}^{n'}\}.
\end{equation*}

We can solve a Hamiltonian structure of (\ref{equ:reduction111}), which is given by
\begin{equation*}
        \{\sigma_{l}^{n},\sigma_{l'}^{n'}\}=
        a_0(n'+l'-n-l)\sigma_{l}^{n}\sigma_{l'}^{n'}, \quad\text{where} \quad a_0 \in \mathbb{Q}.
\end{equation*}
Note that we cannot find the second Hamiltonian structure for \eqref{equ:reduction111} by using the above method. We conjecture that  the equation \eqref{equ:reduction111} possesses the second Hamiltonian structure which takes more complicated form. 
\end{eg}

\begin{eg}There is an example with the Hamiltonian structure which has more degree of freedom. Given the following quiver and the composition of mutations at $\{\sideset{^{n}}{^{m}_{l}}{\mathop{\tau}}|2n-2m+l=-2\}$. 

\resizebox{\textwidth}{!}{
    \begin{tikzcd}[ampersand replacement=\&]
    	\&\&\&\& {\sideset{^{2}}{^{1}_{-4}}{\mathop{\tau}}} \& {\sideset{^{2}}{^{1}_{-3}}{\mathop{\tau}}} \& {\sideset{^{2}}{^{1}_{-2}}{\mathop{\tau}}} \& {\sideset{^{2}}{^{1}_{-1}}{\mathop{\tau}}} \& {\sideset{^{2}}{^{1}_{0}}{\mathop{\tau}}} \\
	\& {\sideset{^{3}}{^{0}_{-5}}{\mathop{\tau}}} \& {\sideset{^{3}}{^{0}_{-4}}{\mathop{\tau}}} \&\&\&\& {\sideset{^{1}}{^{1}_{-2}}{\mathop{\tau}}} \& {\sideset{^{1}}{^{1}_{-1}}{\mathop{\tau}}} \& {\sideset{^{1}}{^{1}_{0}}{\mathop{\tau}}} \& {\sideset{^{1}}{^{1}_{1}}{\mathop{\tau}}} \& {\sideset{^{1}}{^{1}_{2}}{\mathop{\tau}}} \\
	\& {\sideset{^{2}}{^{0}_{-5}}{\mathop{\tau}}} \& {\sideset{^{2}}{^{0}_{-4}}{\mathop{\tau}}} \& {\sideset{^{2}}{^{0}_{-3}}{\mathop{\tau}}} \& {\sideset{^{2}}{^{0}_{-2}}{\mathop{\tau}}} \&\&\&\& {\sideset{^{0}}{^{1}_{0}}{\mathop{\tau}}} \& {\sideset{^{0}}{^{1}_{1}}{\mathop{\tau}}} \& {\sideset{^{0}}{^{1}_{2}}{\mathop{\tau}}} \\
	\&\& {\sideset{^{1}}{^{0}_{-4}}{\mathop{\tau}}} \& {\sideset{^{1}}{^{0}_{-3}}{\mathop{\tau}}} \& {\sideset{^{1}}{^{0}_{-2}}{\mathop{\tau}}} \& {\sideset{^{1}}{^{0}_{-1}}{\mathop{\tau}}} \& {\sideset{^{1}}{^{0}_{0}}{\mathop{\tau}}} \&\&\&\& {\sideset{^{-1}}{^{1}_{2}}{\mathop{\tau}}} \\
	{\sideset{^{2}}{^{-1}_{-4}}{\mathop{\tau}}} \&\&\&\& {\sideset{^{0}}{^{0}_{-2}}{\mathop{\tau}}} \& {\sideset{^{0}}{^{0}_{-1}}{\mathop{\tau}}} \& {\sideset{^{0}}{^{0}_{0}}{\mathop{\tau}}} \& {\sideset{^{0}}{^{0}_{1}}{\mathop{\tau}}} \& {\sideset{^{0}}{^{0}_{2}}{\mathop{\tau}}} \\
	{\sideset{^{1}}{^{-1}_{-4}}{\mathop{\tau}}} \& {\sideset{^{1}}{^{-1}_{-3}}{\mathop{\tau}}} \& {\sideset{^{1}}{^{-1}_{-2}}{\mathop{\tau}}} \&\&\&\& {\sideset{^{-1}}{^{0}_{0}}{\mathop{\tau}}} \& {\sideset{^{-1}}{^{0}_{1}}{\mathop{\tau}}} \& {\sideset{^{-1}}{^{0}_{2}}{\mathop{\tau}}} \& {\sideset{^{-1}}{^{0}_{3}}{\mathop{\tau}}} \\
	{\sideset{^{0}}{^{-1}_{-4}}{\mathop{\tau}}} \& {\sideset{^{0}}{^{-1}_{-3}}{\mathop{\tau}}} \& {\sideset{^{0}}{^{-1}_{-2}}{\mathop{\tau}}} \& {\sideset{^{0}}{^{-1}_{-1}}{\mathop{\tau}}} \& {\sideset{^{0}}{^{-1}_{0}}{\mathop{\tau}}} \&\&\&\& {\sideset{^{-2}}{^{0}_{2}}{\mathop{\tau}}} \& {\sideset{^{-2}}{^{0}_{3}}{\mathop{\tau}}} \\
	\&\& {\sideset{^{-1}}{^{-1}_{-2}}{\mathop{\tau}}} \& {\sideset{^{-1}}{^{-1}_{-1}}{\mathop{\tau}}} \& {\sideset{^{-1}}{^{-1}_{0}}{\mathop{\tau}}} \& {\sideset{^{-1}}{^{-1}_{1}}{\mathop{\tau}}} \& {\sideset{^{-1}}{^{-1}_{2}}{\mathop{\tau}}}
	\arrow[from=2-9, to=2-8]
	\arrow[from=2-10, to=2-9]
	\arrow[from=3-10, to=2-10]
	\arrow[from=3-11, to=3-10]
	\arrow[from=3-9, to=2-9]
	\arrow[from=3-10, to=3-9]
	\arrow[from=2-9, to=3-10]
	\arrow[from=3-9, to=5-8]
	\arrow[from=4-7, to=3-9]
	\arrow[from=2-8, to=4-7]
	\arrow[from=4-7, to=4-6]
	\arrow[from=5-8, to=5-7]
	\arrow[from=5-7, to=4-7]
	\arrow[from=5-6, to=4-6]
	\arrow[from=5-7, to=5-6]
	\arrow[from=4-6, to=5-7]
	\arrow[from=4-6, to=4-5]
	\arrow[from=2-8, to=1-8]
	\arrow[from=1-8, to=1-7]
	\arrow[from=4-5, to=3-5]
	\arrow[from=4-11, to=3-11]
	\arrow[from=4-5, to=4-4]
	\arrow[from=3-5, to=3-4]
	\arrow[from=4-4, to=3-4]
	\arrow[from=4-4, to=6-3]
	\arrow[from=6-3, to=6-2]
	\arrow[from=6-2, to=6-1]
	\arrow[from=6-1, to=5-1]
	\arrow[from=3-4, to=3-3]
	\arrow[from=3-3, to=2-3]
	\arrow[from=2-3, to=2-2]
	\arrow[from=3-4, to=4-5]
	\arrow[from=5-7, to=4-7]
	\arrow[from=6-9, to=6-8]
	\arrow[from=6-8, to=5-8]
	\arrow[from=6-9, to=5-9]
	\arrow[from=3-3, to=3-2]
	\arrow[from=3-2, to=2-2]
	\arrow[from=7-3, to=6-3]
	\arrow[from=7-4, to=7-3]
	\arrow[from=7-5, to=7-4]
	\arrow[from=8-5, to=7-5]
	\arrow[from=8-6, to=8-5]
	\arrow[from=8-7, to=8-6]
	\arrow[from=3-2, to=5-1]
	\arrow[from=5-6, to=7-5]
	\arrow[from=6-8, to=8-7]
	\arrow[from=2-3, to=1-5]
	\arrow[from=1-5, to=3-4]
	\arrow[from=3-4, to=1-6]
	\arrow[from=1-6, to=3-5]
	\arrow[from=1-6, to=1-5]
	\arrow[from=1-7, to=1-6]
	\arrow[from=4-6, to=2-8]
	\arrow[from=2-7, to=4-6]
	\arrow[from=3-5, to=2-7]
	\arrow[from=2-7, to=1-7]
	\arrow[from=2-8, to=2-7]
	\arrow[from=5-9, to=5-8]
	\arrow[from=5-8, to=3-10]
	\arrow[from=3-10, to=5-9]
	\arrow[from=5-9, to=4-11]
	\arrow[from=6-10, to=6-9]
	\arrow[from=7-10, to=6-10]
	\arrow[from=6-10, to=4-11]
	\arrow[from=2-2, to=3-3]
	\arrow[from=5-8, to=6-9]
	\arrow[from=1-7, to=2-8]
	\arrow[from=5-6, to=5-5]
	\arrow[from=5-5, to=4-5]
	\arrow[from=6-7, to=5-7]
	\arrow[from=4-4, to=4-3]
	\arrow[from=4-3, to=3-3]
	\arrow[from=5-7, to=6-8]
	\arrow[from=4-5, to=5-6]
	\arrow[from=3-3, to=4-4]
	\arrow[from=2-9, to=1-9]
	\arrow[from=1-9, to=1-8]
	\arrow[from=1-8, to=2-9]
	\arrow[from=3-11, to=2-11]
	\arrow[from=2-11, to=2-10]
	\arrow[from=2-10, to=3-11]
	\arrow[from=7-5, to=6-7]
	\arrow[from=6-7, to=8-6]
	\arrow[from=8-6, to=6-8]
	\arrow[from=7-10, to=7-9]
	\arrow[from=7-9, to=6-9]
	\arrow[from=8-7, to=7-9]
	\arrow[from=7-4, to=5-6]
	\arrow[from=5-5, to=7-4]
	\arrow[from=6-3, to=5-5]
	\arrow[from=8-5, to=8-4]
	\arrow[from=8-4, to=7-4]
	\arrow[from=8-4, to=8-3]
	\arrow[from=8-3, to=7-3]
	\arrow[from=7-3, to=8-4]
	\arrow[from=7-4, to=8-5]
	\arrow[from=5-1, to=4-3]
	\arrow[from=4-3, to=6-2]
	\arrow[from=6-2, to=4-4]
	\arrow[from=7-3, to=7-2]
	\arrow[from=7-2, to=6-2]
	\arrow[from=7-2, to=7-1]
	\arrow[from=7-1, to=6-1]
	\arrow[from=6-1, to=7-2]
	\arrow[from=6-2, to=7-3]
	\arrow[from=6-8, to=6-7]
	\arrow[from=6-9, to=7-10]
    \end{tikzcd}
}
Reduce this quiver in $(l,m,n)=(2,1,0)$ direction and get the new quiver:

\resizebox{0.9\textwidth}{!}{
    \begin{tikzcd}[ampersand replacement=\&]
	{\sigma_{-2}^{2}} \& {\sigma_{-1}^{2}} \& {\sigma_{0}^{2}} \\
	{\sigma_{-2}^{1}} \& {\sigma_{-1}^{1}} \& {\sigma_{0}^{1}} \& {\sigma_{1}^{1}} \& {\sigma_{2}^{1}} \\
	\&\& {\sigma_{0}^{0}} \& {\sigma_{1}^{0}} \& {\sigma_{2}^{0}} \& {\sigma_{3}^{0}} \& {\sigma_{4}^{0}} \\
	\&\&\&\& {\sigma_{2}^{-1}} \& {\sigma_{3}^{-1}} \& {\sigma_{4}^{-1}}
	\arrow[from=2-5, to=2-4]
	\arrow[from=3-6, to=3-5]
	\arrow[from=3-5, to=2-5]
	\arrow[from=3-4, to=2-4]
	\arrow[from=3-5, to=3-4]
	\arrow[from=2-4, to=3-5]
	\arrow[from=2-4, to=2-3]
	\arrow[from=2-3, to=1-3]
	\arrow[from=2-3, to=2-2]
	\arrow[from=1-3, to=1-2]
	\arrow[from=2-2, to=1-2]
	\arrow[from=1-2, to=1-1]
	\arrow[from=1-2, to=2-3]
	\arrow[from=3-5, to=2-5]
	\arrow[from=4-6, to=3-6]
	\arrow[from=4-7, to=3-7]
	\arrow[from=3-7, to=3-6]
	\arrow[from=3-6, to=4-7]
	\arrow[from=3-4, to=3-3]
	\arrow[from=3-3, to=2-3]
	\arrow[from=4-5, to=3-5]
	\arrow[from=2-2, to=2-1]
	\arrow[from=2-1, to=1-1]
	\arrow[from=3-5, to=4-6]
	\arrow[from=2-3, to=3-4]
	\arrow[from=1-1, to=2-2]
	\arrow[from=4-6, to=4-5]
	\arrow[from=1-3, to=2-1]
	\arrow[curve={height=-24pt}, from=2-4, to=2-2]
	\arrow[curve={height=-24pt}, from=2-2, to=2-5]
	\arrow[from=2-5, to=3-3]
	\arrow[curve={height=-24pt}, from=3-3, to=3-6]
	\arrow[curve={height=-24pt}, from=3-4, to=3-7]
	\arrow[curve={height=-24pt}, from=3-6, to=3-4]
	\arrow[from=3-7, to=4-5]
	\arrow[from=4-7, to=4-6]
	\arrow[curve={height=-24pt}, from=2-1, to=2-4]
\end{tikzcd}
}

This quiver corresponds to the partial difference 
\begin{equation}\label{equ:reduction210}
    \sigma_{l+3}^{n+1}=\frac{\sigma_{l+1}^{n} \sigma_{l+2}^{n+1}+\sigma_l^{n+1} \sigma_{l+3}^{n}}{\sigma_{l}^n}.
\end{equation}

We give the conditions which the Hamiltonian structure $P=(p_{ij})$ of (\ref{equ:reduction210}) should satisfy:

\begin{equation*}
    \begin{split}
        &p_{i,5k-3}-p_{i,5k-1}-p_{i,5k+1}+p_{i,5k+3}=0,\\
        -&p_{i,5k-3}+p_{i,5k-2}+p_{i,5k}-p_{i,5k+2}-p_{i,5k+3}+p_{i,5k+4}=0,\\
        -&p_{i,5k-2}+p_{i,5k-1}+p_{i,5k+1}-p_{i,5k+3}-p_{i,5k+5}+p_{i,5k+6}=0,\\
        -&p_{i,5k}+p_{i,5k+1}+p_{i,5k+2}-p_{i,5k+4}-p_{i,5k+6}+p_{i,5k+7}=0,\\
        &p_{i,5k+1}-p_{i,5k+3}-p_{i,5k+5}+p_{i,5k+7}=0
    \end{split}
\end{equation*}
and
\begin{equation*}
    \begin{array}{ll}
        p_{i,5k-1}+p_{i,5k+1}-p_{i,5k}=p_{i-1,5k-1}\quad\quad&i \equiv 1,2,3,4 \quad\text{mod } 5,   \\
        p_{i,5k}+p_{i,5k-2}-p_{i,5k-1}=p_{i+1,5k} &  i \equiv 0,1,2,3 \quad\text{mod } 5,   \\
        p_{i,j}=p_{i-1,j-1} & i,j \not\equiv 0 \quad\text{mod } 5. 
    \end{array}
\end{equation*}
and
\begin{equation*}
	\{\Lambda_{n}\sigma_{l}^{n},\Lambda_{n}\sigma_{l'}^{n'}\}=\Lambda_{n}\{\sigma_{l}^{n},\sigma_{l'}^{n'}\}.
\end{equation*}

For any integer $n$, there exists a unique integer $k$ such that $n-3k \in \{-1,0,1\}$. For fixed $a_0,b_0 \in \mathbb{Q} $, define a map  
$f:\mathbb{Z}\rightarrow \mathbb{Q}$:
\begin{equation*}
    f(n)=(n-2k)a_0 + k b_0.
\end{equation*}

Then the general Hamiltonian structure of (\ref{equ:reduction111}) is given by:
\begin{equation*}
    \{\sigma_{l}^{n},\sigma_{l'}^{n'}\}=
        (q_{n-n'}+ f(n'+l'-n-l))\sigma_{l}^{n}\sigma_{l'}^{n'},
\end{equation*}
where $q_{n-n'} \in \mathbb{Q}$ satisfying $q_{k}=-q_{-k}$. Actually, this Poisson bracket gives a tri-Hamiltonian structure of \eqref{equ:reduction210}.
\end{eg}
    


\end{document}